\documentclass[runningheads]{llncs}

\usepackage{color}
\usepackage{amsfonts}
\usepackage{appendix}
\usepackage{booktabs} 
\usepackage{cite}
\usepackage{tabularx}
\usepackage{makecell}
\usepackage{verbatim}
\usepackage{amsgen,amsmath,amstext,amsbsy,amsopn,amssymb,amsthm}
\usepackage{cancel}
\usepackage{xspace}
\usepackage{xcolor}
\usepackage{graphicx}
\usepackage{algorithm}
\usepackage{enumitem}
\usepackage{hyperref}
\usepackage[noend]{algpseudocode}
\usepackage{algorithmicx}
\algnewcommand{\LeftComment}[1]{{\color{teal}\textbf{\Statex \(\triangleright\) #1}}} 
\algdef{SE}[AS]{As}{EndAs}{\textbf{as}\ }[1]{}%
\algdef{SE}[DOWHILE]{Do}{Until}{\textbf{loop forever}}[1]{\textbf{until}\ #1}%
\algdef{SxnE}[FOR]{ForInitialize}{EndForInitialize}[1]{\textbf{for every}\ #1,\ \textbf{initialize}:}
\algdef{SxnE}[IF]{Upon}{EndUpon}[1]{\textbf{upon receiving}\ #1\ \algorithmicdo}
\algdef{SxnE}[IF]{Init}{EndInit}{\textbf{Local variables initialization:}}

\newcommand{\Z}{{\mathcal{Z}}}
\newcommand{\A}{{\mathcal{A}}}
\newcommand{\N}{{\mathcal{N}}}
\newcommand{\xw}[1] {{\textcolor{blue}{XW: #1}}}

\renewcommand{\paragraph}[1]{\noindent {\bf #1}}
\usepackage{url}
\usepackage{etoolbox}

\newtoggle{arxiv}
\toggletrue{arxiv}

\begin{document}

\title{Blockchain CAP Theorem Allows User-Dependent Adaptivity and Finality}

\author{Suryanarayana Sankagiri\inst{1}\thanks{The first two authors contributed equally to this work.}
\and
Xuechao Wang\inst{1}{$^\star$}
\and
Sreeram Kannan \inst{2}
\and 
 Pramod Viswanath\inst{1}
}

\authorrunning{S.\ Sankagiri et al.}

\institute{University of Illinois, Urbana-Champaign IL 61801  \and
University of Washington at Seattle, WA}

\maketitle
\begin{abstract}
    Longest-chain blockchain protocols, such as Bitcoin, guarantee liveness even when the number of actively participating users is variable, i.e., they are adaptive. However, they are not safe under network partitions, i.e., they do not guarantee finality. On the other hand, classical blockchain protocols, like PBFT, achieve finality but not adaptivity. Indeed, the CAP theorem in the context of blockchains asserts that no protocol can simultaneously offer both adaptivity and finality. 

    We propose a new blockchain protocol, called the checkpointed longest chain, that offers individual users the choice between finality and adaptivity instead of imposing it at a system-wide level. This protocol's salient feature is that it supports two distinct confirmation rules: one that guarantees adaptivity and the other finality. The more optimistic adaptive rule always confirms blocks that are marked as finalized by the more conservative rule, and may possibly confirm more blocks during variable participation levels. Clients (users) make a local choice between the confirmation rules as per their personal preference, while miners follow a fixed block proposal rule that is consistent with both confirmation rules. The proposed protocol has the additional benefit of intrinsic validity: the finalized blocks always lie on a single blockchain, and therefore miners can attest to the validity of transactions while proposing blocks. Our protocol builds on the notion of a finality gadget, a popular technique for adding finality to longest-chain protocols.
\end{abstract}
\section{Introduction} \label{sec:intro}

The longest-chain protocol, introduced by Nakamoto in Bitcoin \cite{bitcoin}, is the prototypical example of a blockchain protocol that operates in a permissionless setting. Put differently, the longest-chain protocol is {\em adaptive}: it remain safe and live irrespective of the number of active participants (nodes) in the system, as long as the fraction of adversarial nodes among the active ones is less than a half. Adaptivity (also known as dynamic availability) is a desirable property in a highly decentralized system such as a cryptocurrency. The downside of this protocol is that they are insecure during prolonged periods of network partition (asynchrony). This is unavoidable; miners in disconnected portions of the network keep extending their blockchains separately, unaware of the other chains. When synchrony resumes, one of the chains wins, which implies that blocks on the other chains get unconfirmed. These features are present not just in the longest-chain protocol, but also in other protocols that are derived from it, e.g., Prism \cite{bagaria2019prism}.

In stark contrast, committee-based consensus protocols like PBFT \cite{castro1999practical} and Hotstuff \cite{yin2019hotstuff} offer strong {\em finality} guarantees. These protocols remain safe even during periods of asynchrony, and regain liveness when synchrony resumes. Finality is also a desirable property for blockchains, as they may operate under conditions where synchrony cannot be guaranteed at all times. However, all finality-guaranteeing protocols come with a caveat: they are built for the permissioned setting. These protocols make progress only when enough number of votes have been accrued for each block. If a significant fraction of nodes become inactive (go offline), the protocol stalls completely. This prevents them from being adaptive. 

We posit that the two disparate class of protocols is a consequence of the CAP theorem, a famous impossibility result in distributed systems. The theorem states that in the presence of a {\em network partition}, a distributed system cannot guarantee both {\em consistency} (safety) and {\em availability} (liveness) \cite{gilbert2002brewer, gilbert2012perspectives}. The theorem has led to a classification of system designs, on the basis of whether they favor liveness or safety during network partitions. Blockchains, being distributed systems, inherit the trade-offs implicated by the CAP theorem. In particular, we see that the longest-chain class of protocols favor liveness while the committee-based protocols favor safety.

Recently, Lewis-Pye and Roughgarden \cite{lewis2020resource} prove a CAP theorem for blockchains that highlights the adaptivity-finality trade-off explicitly. They show that ``a fundamental dichotomy holds between protocols (such as Bitcoin) that are adaptive, in the sense that they can function given unpredictable levels of participation, and protocols (such as Algorand) that have certain finality properties". The essence of this impossibility result is the following: it is difficult to distinguish network asynchrony from a reduced number of participants in the blockchain system. Therefore, a protocol is bound to behave similarly under both these conditions. In particular, adaptive protocols must continue to extend blockchains during asynchrony (thereby compromising finality), while finality-guaranteeing protocols must stall under reduced participation (which means they cannot have adaptivity).

In this paper, we investigate whether the aforementioned trade-off between adaptivity and finality can be resolved at a user level, rather than at a system-wide level. In particular, we seek to build a blockchain system wherein all (honest) nodes follow a {\em common} block proposing mechanism, but different nodes can choose between two {\em different confirmation rules}. Under appropriate bounds on adversarial participation, one rule must guarantee adaptivity, while the other {must guarantee finality}. When the system is operating under desirable conditions, i.e., a large enough fraction of nodes are active and the network is synchronous, the blocks confirmed by both rules must coincide. Protocols with such a dual-confirmation rule (aka dual-ledger) design are termed as ``Ebb-and-flow" protocols in \cite{neu2020ebb}. Such a design would be of interest in blockchains for many reasons. For example, for low-value transactions such as buying a coffee, a node may prefer liveness over safety, whereas for high-value transactions, it is natural to choose safety over liveness. Moreover, such a trade-off allows each node to make their own assumptions about the state of the network and choose a confirmation rule appropriately. 

At a high level, we are inspired in our formulation from analogous designs to adapt the CAP theorem in practical distributed system settings (Section 4 of \cite{gilbert2012perspectives}). More concretely, dual-ledger designs do not fall under the purview of the CAP theorem of \cite{lewis2020resource}; the formulation assumes a single confirmation rule. A second motivation is that a variety of {\em finality gadgets} (in combination with a blockchain protocol) may also be viewed as providing a user-dependent dual ledger option. We elaborate on some recent proposals in \cite{DinsdaleYoung2020AfgjortAP,stewart2020grandpa, karakostas2020securing} in \S\ref{sec:related}.

\subsection{Our contributions} Our main contribution is to propose a new protocol, called the \textit{checkpointed longest chain} protocol, which offers each node in the same blockchain system a choice between two different confirmation rules that have different adaptivity-finality trade-offs. 
\begin{itemize}
    \item {\bf Block Proposal}. Just as in the longest chain protocol, honest miners build new blocks by extending the chain that they currently hold. In addition, some honest users participate in a separate \textit{checkpointing protocol}, that marks certain blocks as checkpoints at regular intervals. An honest user adopts the \textit{longest chain that contains the latest checkpoint}.
    \item {\bf Confirmation rules}. The protocol's \textit{adaptivity-guaranteeing} confirmation rule is simply the $k$-deep rule. An honest user confirms a block if it sees $k$ blocks below it, for an appropriate choice of $k$. The protocol's \textit{finality-guaranteeing} rule is for honest users to treat all blocks in its chain up to the last checkpointed block as confirmed.
\end{itemize}
The protocol is designed such that new blocks continue to be mined at increasing heights even if the participation level is low, but new checkpoints appear only if there is sufficient participation. This illustrates the adaptivity-guaranteeing property of the $k$-deep rule. On the other hand, checkpoints are guaranteed to be on a single chain irrespective of network conditions, while the longest chain containing the latest checkpoint may keep alternating between divergent chains. This implies the finality property of the checkpointing rule.

Below, we highlight the key design principles of our protocol, followed by the security guarantees.

\subsubsection{Validity of Blockchains} Our protocol keeps intact the intrinsic \textit{validity} of blockchains. Roughly, validity means that the set of blocks that are confirmed all lie on a single, monotonically increasing chain. Thus, the validity of a transaction can be inferred from just the blockchain leading up to the particular block. More precisely, an honest user that is constructing the block has the assurance that if the block is confirmed, all transactions in the block are confirmed, which means it accrues all the transaction rewards. Moreover, a user can infer that a certain transaction is confirmed simply by knowing that the transaction was included in a particular block and that the block was confirmed (as per either rule), without knowledge of the contents of other blocks. These properties are recognized to be important for the blockchain to be incentive-compatible and to enable light clients respectively. Further, validity of blockchains offers spam-resistance and is compatible with existing sharding designs. Although validity of blockchains is a common feature of a majority of protocols, recent high performance designs crucially decouple validation from consensus {\cite{bagaria2019prism,fitzi2018parallel}}. The Snap-and-Chat design of \cite{neu2020ebb} for the same problem also does not have this feature.

\begin{figure}[ht]
\begin{centering}
\includegraphics[scale=0.2]{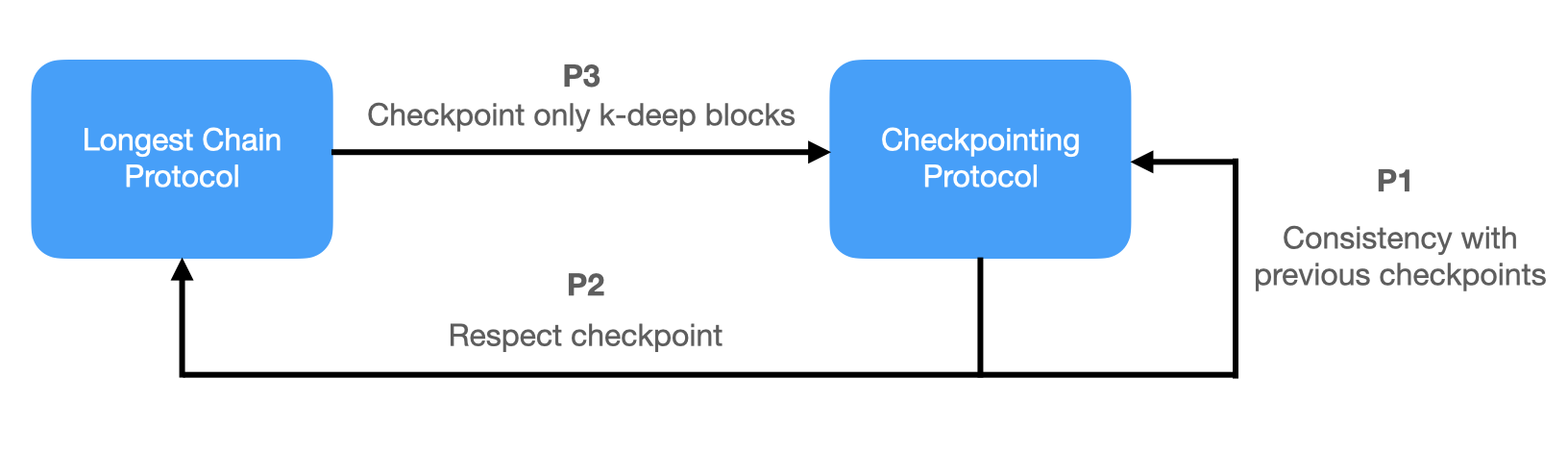}
\caption{The Checkpointed Longest Chain}
\end{centering}
\label{fig:1}
\end{figure}

\subsubsection{Interaction between Checkpointing and Longest Chain Protocol} In our design, the longest chain rule (i.e., the block proposal rule) and the checkpointing rule satisfy some constraints with respect to each other. The interlocking structure of constraints between the checkpointing protocol and the longest chain protocol is illustrated in Figure~\ref{fig:1} and detailed below.
\begin{itemize}
    \item \textbf{P1: Consistency with previous checkpoints.} The sequence of checkpoints must be on a chain. This places a self-consistency condition on the checkpointing protocol.
    \item \textbf{P2: Checkpointed longest chain rule.} The longest chain protocol respect the previous checkpoints. Honest users build adopt the longest chain that contains the latest checkpoint and mine new blocks at the tip of this chain.
    \item \textbf{P3: Checkpoints are deep enough.} {The checkpointing protocol should only checkpoint blocks that are sufficiently deep in some honestly held chain.} Furthermore, such a chain should be made available along with the checkpoint message.
\end{itemize}  

All these three conditions are required to achieve our goals. Condition {\bf P1} is a requisite for a validity preserving protocol (in particular, a validity preserving, finality-guaranteeing confirmation rule). The condition {\bf P2} says that honest nodes should adopt the longest checkpointed chain, rather than the longest chain. Given {\bf P1}, the  condition {\bf P2} is required to ensure new checkpoints are produced after a period of asynchrony. Without this condition {\bf P2}, the block proposal rule is simply the longest chain rule. During asynchrony, the longest chain may be one that does not include some of the checkpoints, and there is no correction mechanism to ever bring the longest chain downstream of the last checkpoint. Condition {\bf P3} ensures that during synchrony, any block that is eventually checkpointed would be a part of all honestly held chains. Thus, the dynamics of the protocol would be as if nodes are simply following the longest chain rule, and the known security guarantees would apply. If no such condition is placed, the checkpointing protocol could checkpoint an arbitrary block that possibly forks from the main chain by a large margin. The rule {\bf P2} would force honest nodes to switch to this fork, thereby violating safety of the $k$-deep rule.

\subsubsection{The Checkpointing Protocol} The checkpointing protocol in our design is a BFT consensus protocol. It is a slight modification of Algorand BA, that is presented in \cite{chen2018algorand}. The protocol is extended from a single-iteration byzantine agreement protocol to a multi-iteration checkpointing protocol. In each iteration, nodes run the checkpointing protocol to checkpoint a new block, and {\bf P1} guarantees that the sequence of checkpoints lie on a single chain. Coupled with the chain adoption rule (longest checkpointed chain), these checkpoints offer deterministic finality and safety against network partitions.

We state the safety guarantee of the checkpointing protocol as a basic checkpointing property (CP) as below:
\begin{itemize}
    \item \textbf{CP0: Safety.} All honest users checkpoint the same block in one iteration of the checkpointing protocol, even during network partition. This checkpoint lies on the same chain as all previous checkpoints.
\end{itemize}

While one might consider different protocols for the checkpointing mechanism, the protocol must further satisfy some key properties during periods of synchrony.
\begin{itemize}
    \item \textbf{CP1: Recency condition.} If a new block is checkpointed at some time, it must have been in a chain held by an honest user at some point in the \textit{recent past}.
    \item \textbf{CP2: Gap in checkpoints.} The interval between successive checkpoints must be large enough to allow {\bf CP1} to hold.
    \item \textbf{CP3: Conditional liveness.} If all honest nodes hold chains that have a common prefix (all but the last few blocks are common), then a new checkpoint will appear within a certain bounded time (i.e., there is an upper bound on the interval between two successive checkpoints).
\end{itemize}

{\bf CP1} helps in bounding the extent to which an honest user may have to drop honest blocks in its chain when it adopts a new checkpoint. Without this condition, (i.e., if arbitrarily old checkpoints can be issued), honest users may have to let go of many honestly mined blocks. This causes a loss in mining power, which we refer to as \textit{bleeding}. Our choice of Algorand BA was made because it has this property. Other natural candidates such as PBFT \cite{castro1999practical} and HotStuff \cite{yin2019hotstuff}, do not have this property. We discuss this in detail in Appendix \iftoggle{arxiv}{\ref{sec:discussion}}{D of the full paper \cite{sankagiri2020blockchain}}. {\bf CP2} ensures that the loss in honest mining power due to the above mechanism is limited. This condition can be incorporated by design. {\bf CP3} is a liveness property of the checkpointing protocol. It ensures new checkpoints appear at frequent intervals, leading to liveness of the finality-guaranteeing confirmation rule.

\subsubsection{Security Guarantees (Informal)} We show the following guarantees for proof-of-work longest chain along with Algorand BA (augmented with the appropriate validity condition).
\begin{enumerate}
    \item The $k$-deep rule is safe and live if the network is synchronous from the beginning and the fraction of adversaries among online users is less than a half.
    \item The checkpointing protocol, which is safe by design, is also live soon after the network partition is healed under the partially synchronous model. 
    \item The ledger from the checkpoint rule is a prefix of the ledger from the $k$-deep rule from the point of view of any user.
\end{enumerate}

\subsubsection{Proof Technique} We prove these security guarantees using the following strategy. First, we show that if all honest users hold chains that obey the $k$-common prefix condition for an extended period of time, then new blocks get checkpointed at regular intervals, and these blocks are part of the common prefix of the honestly held chains. This tells us that under conditions in which the vanilla longest chain rule is secure, the checkpointed longest chain rule has exactly the same dynamics. It therefore inherits the same security properties of the longest chain rule. Secondly, we show that once sufficient time has passed after a network partition is healed, the chains held by the honest user under the checkpointed longest chain rule are guaranteed to have the common prefix property. Coupled with the first result, we infer that new checkpoints will eventually be confirmed thereby proving liveness of the checkpoint-based confirmation rule under partial synchrony.

\subsubsection{Outline} We review related works in \S\ref{sec:related} and place our results in the context of several recent works on the same topic as this paper. The  similarities and differences in the techniques in our work with closely related works is pictorially illustrated in Figure~\ref{fig:related}. We state our network and security models formally in \S\ref{sec:model}. In \S\ref{sec:protocol}, we describe the checkpointed longest chain protocol, highlighting the roles of the miners and precisely stating the two confirmation rules. For clarity, we describe the checkpointing protocol as a black box with certain properties here; we give the full protocol in Appendix \ref{sec:checkpointing}. We state our main theorem, concerning the security guarantees of our protocol, in \S\ref{sec:mainresult}. In the rest of the section, we give a proof sketch of the theorem. The formal proofs are given \iftoggle{arxiv}{in Appendices \ref{sec:checkpointing_property} and \ref{sec:checkpointed_LC}}{in the full version of our paper \cite{sankagiri2020blockchain} in Appendices B and C}. We conclude the paper with a discussion and pointers for future work in \S\ref{sec:conclusion}.
\section{Related Work}
\label{sec:related}
\paragraph{CAP Theorem} The formal connection between the CAP theorem to blockchains is recently made in  \cite{lewis2020resource}, by  providing an abstract framework in which a wide class of blockchain protocols can be placed, including longest chain protocols (both Proof-of-Work based \cite{garay2015bitcoin} and PoS based, e.g., \cite{daian2019snow}) as well as BFT-style protocols (e.g., \cite{chen2016algorand}).  The main result of \cite{lewis2020resource} says that a protocol (containing a  block ``encoding" procedure and a block confirmation rule) which is adaptive (i.e., which remains live in an unsized setting) cannot offer finality (i.e., deterministic safety, even under arbitrary network conditions) and vice-versa. However, the same paper is mute on the topic of whether different block confirmation rules could offer different guarantees, which is the entire focus of this work. We are inspired in our formulation from analogous designs to adapt the CAP theorem in practical distributed system settings (Section 4 of \cite{gilbert2012perspectives}). 


\paragraph{User-Dependent Confirmation Rules} The idea of giving users the option of choosing their own confirmation rule, based on their beliefs about the network conditions and desired security level was pioneered by Nakamoto themself, via choosing the value of $k$ in the $k$-deep confirmation rule. A recent work \cite{malkhi2019flexible} allows users the option of choosing between partially synchronous confirmation rule and a synchronous one. However, {\em neither} of the confirmation rules are adaptive, as the block proposal rule itself is a committee-based one and cannot make progress once the participation is below a required level. Therefore, they do not ``break" the CAP theorem as proposed in \cite{lewis2020resource}. In contrast, our protocol offers users the choice of an adaptive system.

\paragraph{Hybrid Consensus} Hybrid consensus \cite{pass2017hybrid} is a different technique of incorporating a BFT protocol into the longest-chain protocol. In a nutshell, the idea in this design is to use the longest-chain protocol to randomly select a committee, which then executes a BFT protocol to confirm blocks. The committee consists of the miners of the blocks on the longest chain over a shifting interval of constant size, and is thus re-elected periodically. Hybrid consensus addresses the problem of building a responsive protocol (i.e., latency proportional to network delay) in a PoW setting. It does not, however, address the adaptivity-finality dilemma. In fact, as we illustrate here, the protocol offers neither adaptivity nor finality. Once a miner is elected as a committee member, it is obliged to stay active until the period of its committee is over, which compromises on the adaptivity property. As for finality, although the blocks are confirmed by a finality-based protocol, the committee election mechanism compromises the finality guarantee. During asynchrony, nodes cannot achieve consensus on who belongs to the committee, and thus the whole protocol loses safety.

\begin{figure}
\begin{centering}
\includegraphics[scale=0.16]{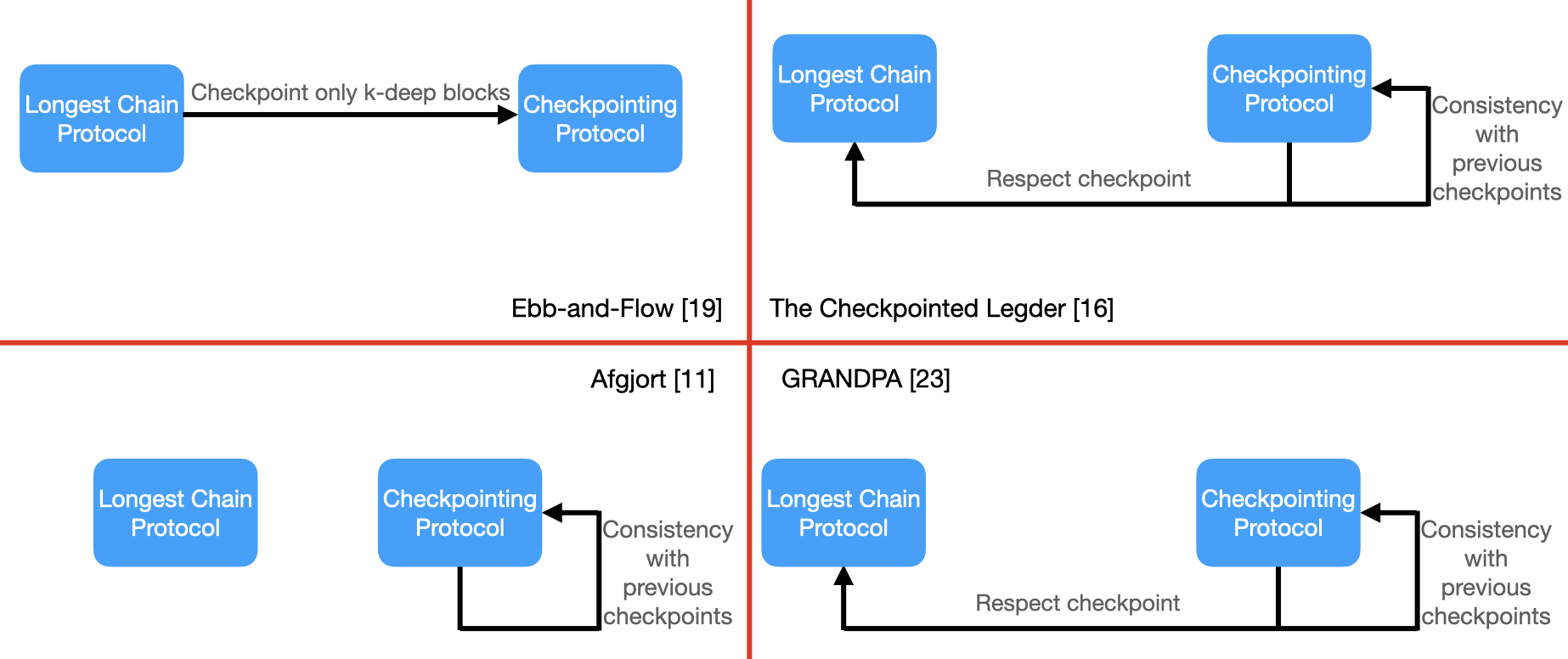}
\caption{Four related works in our framework (counterclockwise): In Ebb-and-Flow \cite{neu2020ebb}, previous checkpoints are not respected by either the longest chain protocol or the checkpointing protocol. Hence checkpoints may not form a single chain and the final ledger is constructed by sanitization, which makes it impossible to have validation before block proposal. Afgjort \cite{DinsdaleYoung2020AfgjortAP} takes an off-the-shelf longest chain protocol and keeps checkpointing blocks on it. Due to lack of interaction between the longest chain protocol and the checkpointing protocol, the protocol may fail in a partially synchronous setting when the security of the blockchain is broken. The Checkpointed Ledger \cite{karakostas2020securing} and GRANDPA \cite{stewart2020grandpa} are very similar to our design; the difference is that the checkpoints don't need to be buried deep enough in the longest chain. This minor change makes their protocols become insecure under the variable participation setting.}
\end{centering}
\label{fig:related}
\end{figure}

\paragraph{Checkpointing} is a protocol run on top of the longest chain protocol (run by largely honest parties) that deterministically marks certain blocks as finalized; a simple form of checkpointing was  pioneered, and maintained until 2014, by Satoshi Nakamoto (presumably honest). In the context of our work, checkpointing provides finality (essentially by fiat) while the longest chain rule ensures adaptivity. 
We can interpret our paper as an attempt to provide appropriate conditions on the interaction between the checkpointing and the longest chain protocol, and also to show how to realize such checkpointing in a distributed manner. We note that \cite{karakostas2020securing} is a  recent work that studies the checkpointed ledger when the longest-chain protocol has super-majority adversaries; this regime is different from that studied in the present paper, where we assume that the longest-chain protocol has majority honest and the checkpointing protocol has $2/3$ honest users. For completeness, we have included the architecture of checkpointed ledger in the comparison in Figure~\ref{fig:related}.

\paragraph{Casper FFG and Gasper} Casper FFG ({\it Casper: the Friendly Finality Gadget}), presented in \cite{buterin2017casper}, pioneered the study of finality gadgets. Casper FFG also introduced the notion of  ``economic security". The finality gadget allows cryptographic proofs of malicious behavior which can be used to disincentivize such behavior. We do not explore this aspect of finality gadgets in our paper. Casper FFG, as presented in \cite{buterin2017casper}, is not a completely specified protocol. A recent follow-up paper called Gasper \cite{buterin2020combining} provides a complete protocol. Gasper can be viewed as a checkpointed longest chain protocol, with the longest chain rule replaced by the {\it Latest Message Driven Greediest Heaviest Observed Subtree} (LMD GHOST) rule and the checkpointing protocol being  Casper FFG.  Gasper does not provide formal security guarantees; in fact, the work by Neu et al. \cite{neu2020ebb} shows a liveness attack on Gasper.

\paragraph{Afgjort}  Afgjort \cite{DinsdaleYoung2020AfgjortAP} is a recent finality gadget proposed by Dinsdale et al., which formally describes the desirable properties that a finality gadget must have, some of them are similar to the ones we require as well. Afgjort has a two-layer design: it takes an off-the-shelf longest chain protocol and adds a finality layer to it. Such a system can work as desired when the network is synchronous, including when the participation levels are variable. However, it is destined to fail in one of two ways in a partially synchronous setting: either the finality gadget stops finalizing new blocks (i.e., violates liveness), or it finalizes ``conflicting" blocks, i.e., blocks in two different chains (i.e., violates validity). Thus Afgjort does not provide the guarantees we seek in this work. More generally, any layer-two finality gadget on top of an adaptive protocol would not meet the desired objectives for the same reason; they would be missing rule {\bf P2} (see Figure \ref{fig:related}).

\paragraph{GRANDPA}  GRANDPA \cite{stewart2020grandpa} is another recent finality gadget, proposed by Stewart et al. Just as in our protocol (and unlike Afgjort), GRANDPA alters the underlying block production mechanism to respect blocks that has been finalized by the finality gadget. GRANDPA's security relies on the assumption that the block-production mechanism (longest-chain rule/GHOST) provides a type of ``conditional eventual consensus". I.e., If nodes keep building on the last finalized block and don’t finalize any new blocks, then eventually they have consensus on a longer chain. This argument, though presumably true, is not a rigorous security analysis; in our work, we prove security for a completely specified model and protocol.  

More importantly, comparing Figures \ref{fig:1} and \ref{fig:related}, we see that GRANDPA does not have property {\bf P3}. This leads to a security vulnerability in the variable participation setting. The issue arises due to the following subtlety: in the variable participation setting, a particular block could be `locked' onto by the checkpointing protocol, but could be checkpointed much later (this does not happen under full participation). A block at the tip of the longest chain could soon be displaced from it by the adversary. If this block is locked while it is on the chain, but checkpointed only when it is out of the chain, the safety of the $k$-deep rule will be violated. A more detailed description of this attack is given in Appendix \iftoggle{arxiv}{\ref{sec:discussion}}{D of the full version of our paper \cite{sankagiri2020blockchain}}. 

\paragraph{Snap-and-Chat protocols} A concurrent work by Neu et al. \cite{neu2020ebb} also solves the adaptivity vs finality dilemma posed by the CAP theorem. (Their work appeared online a few weeks prior to ours).
They introduce a technique, called Snap-and-Chat, to combine a longest chain protocol with any BFT-style one to obtain a protocol that has the properties we desire. More specifically, their protocol produces two different ledgers (ordered list of transactions), one that offers adaptivity and the other that offers finality. A user can choose to pick either one, depending on its belief about the network condition. These ledgers are proven to be consistent with each other (the finality guaranteeing ledger is a prefix of the adaptive one).  While the outcomes of the design of \cite{neu2020ebb} are analogous to ours, the approach adopted and the resulting blockchain protocol are quite different. 
A major advantage of the approach of \cite{neu2020ebb} is its generality: unlike other existing approaches, it offers a blackbox construction and proof which can be used to combine a variety of adaptive and finality preserving protocols. However, the design of \cite{neu2020ebb} has an important caveat: the sequence of blocks that are confirmed in their protocol do not necessarily form a single chain. Indeed, in the partially synchronous setting, the finality-guaranteeing confirmation rule can finalize blocks on two different forks one after the other. The design of \cite{neu2020ebb} overcomes this apparent issue by constructing a ledger after the blocks have been finalized, and ``sanitizing" it at a later step. 

The sanitization step implies that not all transactions in a confirmed block may be part of the ledger: for example, if it is a double spend relative to a transaction occurring earlier in the ledger, such transaction will be removed. Put differently, the validity of a particular transaction in a particular block is decided only once the block is confirmed, and not when the block is proposed. Most practical blockchain systems are designed with coupled validation, i.e., an honest block proposer can ensure that blocks contain only valid transactions. Our approach maintains coupled validity of blockchains and does not have this shortcoming, as highlighted in Section \S\ref{sec:intro}. 
\section{Security Model}
\label{sec:model}

\noindent {\bf Environment.} A blockchain protocol $\Pi$ is directed by an \emph{environment} $\Z(1^\kappa)$, where $\kappa$ is the security parameter, i.e., the length of the hash function output.
This environment
$(i)$ initiates a set of participating nodes $\N$;
$(ii)$ manages a public-key infrastructure (PKI) and assigns each node with a unique cryptographic identity;
$(iii)$ manages nodes through an adversary $\A$ which \emph{corrupts} a subset of nodes before the protocol execution starts;
$(iv)$ manages all accesses of each node from/to the environment including broadcasting and receiving messages.

\noindent {\bf Network model.}  The nodes’ individual timers do not need to be synchronized or almost synchronized. We only require they have the same speed. In our network, we have a variety of messages, including blocks, votes, etc. The following message delay bounds apply to all messages. Secondly, all messages sent by honest nodes are broadcast messages. Thirdly, all honest nodes re-broadcast any message they have heard. The adversary does not suffer any message delay. In general, we assume that the adversary controls the order of delivery of messages, but the end-to-end network delay between any two honest nodes is subject to some further constraints that are specified below. We operate under either one of the two settings specified below: 
\begin{itemize}
    \item {\bf M1 (Partial synchrony model):} A global stabilization time, \textsf{GST}, is chosen by the adversary, unknown to the honest nodes and also to the protocol designer. Before \textsf{GST}, the adversary can delay messages arbitrarily. After \textsf{GST}, all messages sent between honest nodes are delivered within $\Delta$ time. Moreover, messages sent by honest nodes before \textsf{GST} are also delivered by \textsf{GST} + $\Delta$.
    \item {\bf M2 (Synchrony model):} All messages sent from one honest node to another are delivered within $\Delta$ time. 
\end{itemize}

\noindent {\bf Participation model.} We introduce the notion of \textit{sized/unsized} number of nodes to capture the notion that  node network activity (online or offline) varies  over time. It is important to note that we allow for adversarial nodes to go offline as well. A node can come online and go offline at any time during the protocol. An online honest node always executes the protocol faithfully. An offline node, be it honest or adversarial, does not send any messages. Messages that are scheduled to be delivered during the time when the node is offline are delivered immediately after the node comes online again. We consider two different scenarios:
\begin{itemize}
    \item {\bf U1 (Sized setting):} All nodes stay online at all times.
    \item {\bf U2 (Unsized setting):} Nodes can come online and offline at arbitrary times, at the discretion of the adversary, provided a certain minimal number of nodes stay online. 
\end{itemize}

\noindent {\bf Abstract protocol model.} We describe here our protocol in the abstract framework that was developed in \cite{lewis2020resource}. In this framework, a blockchain protocol is specified as a tuple $\Pi = (I,O,C)$, where $I$ denotes the {\it instruction set}, $O$ is an {\it oracle} that abstracts out the leader election process, and $C$ is a confirmation rule (e.g., the $k$-deep confirmation rule). The bounds on adversarial ratios are also specified through $O$. The rationale for adopting this framework is to point out precisely how we circumvent the CAP theorem for blockchains, which is the main result of \cite{lewis2020resource}. 

The theorem states that no protocol $\Pi = (I,O,C)$ can simultaneously offer both finality (safety in the partially synchronous setting) and adaptivity (liveness in the unsized setting). Our checkpointed longest chain protocol is a $4$-tuple $\Pi_{\rm CLC} = (I,O,C_1,C_2)$, with two confirmation rules $C_1$ and $C_2$. This entire protocol can be split at a node level into two protocols: $\Pi_{\rm fin}  = (I,O,C_1)$ and $\Pi_{\rm ada}  = (I,O,C_2)$. The protocol $\Pi_{\rm fin}$ guarantees (deterministic) safety and (probabilistic) liveness under network assumption ${\rm M1}$ and participation level ${\rm U1}$, just as many BFT-type consensus protocols do. It therefore guarantees finality. $\Pi_{\rm ada}$ guarantees safety and liveness under network assumption $M2$ and participation level $U2$, just as longest chain protocols do. It thereby guarantees adaptivity. Finally, both $\Pi_{\rm fin}$ and $\Pi_{\rm ada}$ are safe and live under conditions $M2$ and $U1$; moreover, the set of blocks confirmed by $\Pi_{\rm fin}$ at any time is a subset of those confirmed by $\Pi_{\rm ada}$. A similar property is proven in \cite{neu2020ebb}.

Each node can choose one of the two confirmation rules according to their demands and assumptions. We specify the exact specifications of $I, O, C_1$ and $C_2$ and the associated security properties in the next section.
\section{Protocol Description}
\label{sec:protocol}
\noindent {\bf Nodes in the protocol.} In our protocol, let $\N$ denote the set of all participating nodes. Among these, there are two subsets of participating nodes $\N_1$ and $\N_2$. {We call nodes in $\N_1$ {\em miners}, whose role is to propose new blocks. We call nodes in $\N_2$ {\em checkpointers}, whose role is to vote for blocks on the blockchains they currently hold and mark them as {\em checkpoints}. Note that we allow for any relation among these sets, including the possibility that all nodes play both roles ($\N_1 = \N_2 = \N$), or the two kinds of nodes are disjoint ($\N_1 \cap \N_2 = \phi$)}. The instruction set $I$ and the oracle $O$ is different for miners and checkpointers, and is specified below. Note that in the sized/unsized setting, both $\N_1$ and $\N_2$ are sized/unsized.

Before we specify the protocol, we introduce some terminology pertaining to checkpoints. {A {\em checkpoint certificate} is a set of votes from at least $2/3$ of the checkpointers for a certain block, that certifies that the block is a checkpoint.} When an honest node receives both a checkpoint certificate and a chain that contains the block being checkpointed, we say that the honest node has heard of a checkpoint. We say that a checkpoint \textit{appears at time $t$} if the $t$ is the first time that an honest node hears of it. 

\noindent {\bf $I$ for all nodes.} Every honest node holds a single blockchain at all times, which may be updated upon receiving new messages. They all follow the {\it checkpointed longest chain rule}, which states that a node selects the longest chain \textit{that extends the last checkpointed block} it has heard of so far. Ties are broken by the adversary. To elaborate, a node keeps track of the last checkpoint it has heard of so far. If a node receives a longer chain that includes the last checkpoint, it adopts the received new chain. A node ignores all chains which do not contain the last checkpoint block as per their knowledge. 

\noindent {\bf $I, O$ for miners.} An honest miner adopts the tip of its checkpointed longest chain as the parent block. The miner forms a new block with a {\it digest} of the parent block and all transactions in its buffer, and broadcasts it when it is chosen as a leader by the oracle. The oracle $O$ chooses miners as leaders at random intervals, that can be modeled as a Poisson process with rate $\lambda$. This joint leader election process can be further split into two independent Poisson processes. Let the fraction of online adversarial miners be $\beta < 1/2$. The honest blocks arrive as a Poisson process of rate $(1-\beta)\lambda$, while the adversarial blocks arrive as an independent Poisson process of rate $\beta \lambda$.

\noindent {\bf $I, O$ for checkpointers.} The checkpointers run a multi-iteration Byzantine Agreement (BA) protocol $\Pi_{\rm BA}$ to checkpoint blocks, with each iteration checkpointing one block. Our protocol is a slight variant of the Algorand BA protocol from \cite{chen2018algorand}. The complete protocol is described in Appendix \ref{sec:checkpointing} with a minor modification from Algorand, which is highlighted in red. This modification is made so as to enable a consistency check across iterations without losing liveness (i.e., to get a multi-iteration BA from a single iteration one). 

Briefly, the protocol works as follows. Each iteration is split into {\em periods}. Each period has a unique leader chosen by the oracle. {Nominally, the values on which consensus is to be achieved amongst all checkpointers are the blockchains held by the checkpointers. In practice, the values may be the hashes of the last block of the blockchain.} A key difference from the Algorand protocol is that we allow the checkpointers to change their values at the beginning of each period. The checkpointers aim to achieve consensus amongst these values. The final chain that has been agreed upon is broadcast to all honest nodes (not just checkpointers) together with a certificate. The block that is exactly $k$-deep in this chain is chosen as the checkpoint block in current iteration. Here, $k$ is a parameter of the protocol $\Pi_{\rm CLC}$, and is chosen to be $\Theta(\kappa)$. We call it the \textit{checkpoint depth parameter}. 

\subsubsection{Confirmation rules} We propose the following two confirmation rules for $\Pi_{\rm CLC}$, which have different security guarantees under different assumptions. Either rule can be adopted by any node in the protocol
\begin{itemize}
    \item $C_1$: Confirm the chain up to the last known checkpoint.
    \item $C_2$: Confirm all but the last $k'$ blocks in the checkpointed longest chain, i.e., in the chain that is currently held, where $k' = \Theta(\kappa)$
\end{itemize}
Note that every honest node may choose any value of $k'$ that they wish.

\subsubsection{Intuition behind the checkpointing protocol}
The checkpointing protocol is designed such that under optimistic conditions, a checkpoint is achieved within one period. These optimistic conditions are that the leader of a period is honest, the network is synchronous, and all chains held by honest checkpointers satisfy the common prefix property. The sequence of leaders for this protocol is guaranteed to be chosen in an i.i.d. fashion amongst online checkpointers by the oracle $O$. We assume that the fraction of adversarial checkpointers is $\beta' < 1/3$. Thus, the probability of selecting an adversarial leader at any round is $< 1/3$.

We now state some key properties of the checkpointing protocol. These properties, {\bf CP0, CP1, CP2, CP3} were introduced in \S\ref{sec:intro}  and  are elaborated upon below.
\begin{itemize}
    \item \textbf{CP0: Safety.} All honest nodes checkpoint the same block in one iteration of the checkpointing protocol, even during network partition. This checkpoint lies on the same chain as all previous checkpoints. This is a safety property of the checkpointing protocol, which will be essential in guaranteeing the safety of $\Pi_{\rm fin}$.
    \item \textbf{CP1: Recency condition.} If a new block is checkpointed at some time $t$ by an honest node for the first time, it must have been in a chain held by an honest node at some time $t' \geq t - d$. Here, $d$ is the \textit{recency} parameter of the protocol, and is of the order $O(\sqrt{\kappa} \Delta)$. 
    \item \textbf{CP2: Gap in checkpoints.} If a new block is checkpointed at some time $t$ by an honest node for the first time, then the next iteration of the checkpointing protocol (to decide the next checkpoint) will begin at time $t + e$. Here, $e$ is the \textit{inter-checkpoint interval} and is chosen such that $e \gg d$.
    \item \textbf{CP3: Conditional liveness.} If all honest nodes hold chains that are self-consistent (they satisfy $k$-common prefix property), during an iteration of the checkpointing protocol, then the checkpointing protocol finishes within $O(\Delta)$ time. Thus, in a period where all honest nodes hold chains that are self-consistent, checkpoints appear within $e + O(\Delta)$ time of each other.
\end{itemize}

In Appendix \iftoggle{arxiv}{\ref{sec:checkpointing_property}}{B of the full version \cite{sankagiri2020blockchain}}, we prove that the checkpointing protocol we use (modified Algorand BA) satisfies these properties. As such, any protocol that satisfies {\bf CP0 - CP3} can be used in our design.

\section{Main Result}
\label{sec:mainresult}

To state the main security result of our protocol, recall the notations set in \S\ref{sec:model}. The complete protocol is denoted by $\Pi_{\rm CLC} = (I,O,C_1,C_2)$ and its two user-dependent variations are $\Pi_{\rm fin}  = (I,O,C_1)$ and $\Pi_{\rm ada}  = (I,O,C_2)$. We also recall the notation ${\rm M1, M2}$ for the partially synchronous and synchronous network model, and ${\rm U1}$, ${\rm U2}$ for the sized and unsized participation models. Clearly, ${\rm M1}$ is a more general setting than ${\rm M2}$ and ${\rm U2}$ is a more general setting than ${\rm U1}$. In Theorem \ref{thm:main}, a result stated for a general setting also applies to the more restricted setting, but not vice-versa. 
We first define the notion of safety and liveness that we use for our protocols.
\begin{definition}[Safety]\label{def:safety}
A blockchain protocol $\Pi$ is safe if the set of blocks confirmed during the execution is a non-decreasing set. Put differently, $\Pi$ is safe if a block once confirmed by the confirmation rule remains confirmed for all time thereafter. 
\end{definition}

\begin{definition}[Liveness]\label{def:liveness}
A blockchain protocol $\Pi$ is live if there exists constants $c, c' > 0$ s.t. the number of new honest blocks confirmed in any interval $[r, s]$ is at least $\lfloor c(s-r) - c'\rfloor$.
\end{definition}
We now state our main theorem, concerning the safety and liveness of $\Pi_{\rm fin}$ and $\Pi_{\rm ada}$.
\begin{theorem}
\label{thm:main}
{Assume the fraction of adversarial mining power among total honest mining power is bounded by $\beta < 1/2$. Further, assume the fraction of adversarial checkpointers is always less than $1/3$.} Then, the protocol $\Pi_{\rm CLC}$ has the following security guarantees for an execution that runs for a duration of $T_{\rm max} = O({\rm poly}(\kappa))$ time, where $\kappa$ is the security parameter, :
\begin{itemize}
    \item \textbf{Security of $\Pi_{\rm fin}$:} The protocol $\Pi_{\rm fin}$ is safe, and is live after $O(\textsf{GST} + \kappa)$ time in the setting (${\rm M1, U1}$), except with probability negligible in $\kappa$.
    \item \textbf{Security of $\Pi_{\rm ada}$:} The protocol $\Pi_{\rm ada}$ is safe and live in the setting (${\rm M2, U2}$), except with probability negligible in $\kappa$.
    \item \textbf{Nested Protocols:} {At any time $t$}, the set of blocks confirmed by $\Pi_{\rm ada}$ is a superset of the set of blocks confirmed by $\Pi_{\rm fin}$ in all settings.
\end{itemize}
\end{theorem}
We provide a proof sketch for Theorem \ref{thm:main} below. A detailed proof is relegated to Appendices \iftoggle{arxiv}{\ref{sec:checkpointing_property} and \ref{sec:checkpointed_LC}}{B and C of the full version of our paper \cite{sankagiri2020blockchain}}.

\subsubsection{Proof sketch} 
Our proof for Theorem \ref{thm:main} can be split into two parts. First, Appendix \iftoggle{arxiv}{\ref{sec:checkpointing_property}}{B of the full version \cite{sankagiri2020blockchain}}, we show that our checkpointing protocol $\Pi_{BA}$ satisfies the four checkpointing properties given in \S\ref{sec:protocol}, namely \textbf{CP0, CP1, CP2, CP3}. Then, in Appendix \iftoggle{arxiv}{\ref{sec:checkpointed_LC}}{C of the full version \cite{sankagiri2020blockchain}}, we show that our complete protocol $\Pi_{\rm CLC}$ satisfies the desired security properties if it uses \textit{any} sub-protocol for checkpointing that satisfies the above properties.
\begin{figure}[htbp]
    \centering
    \includegraphics[width=\linewidth]{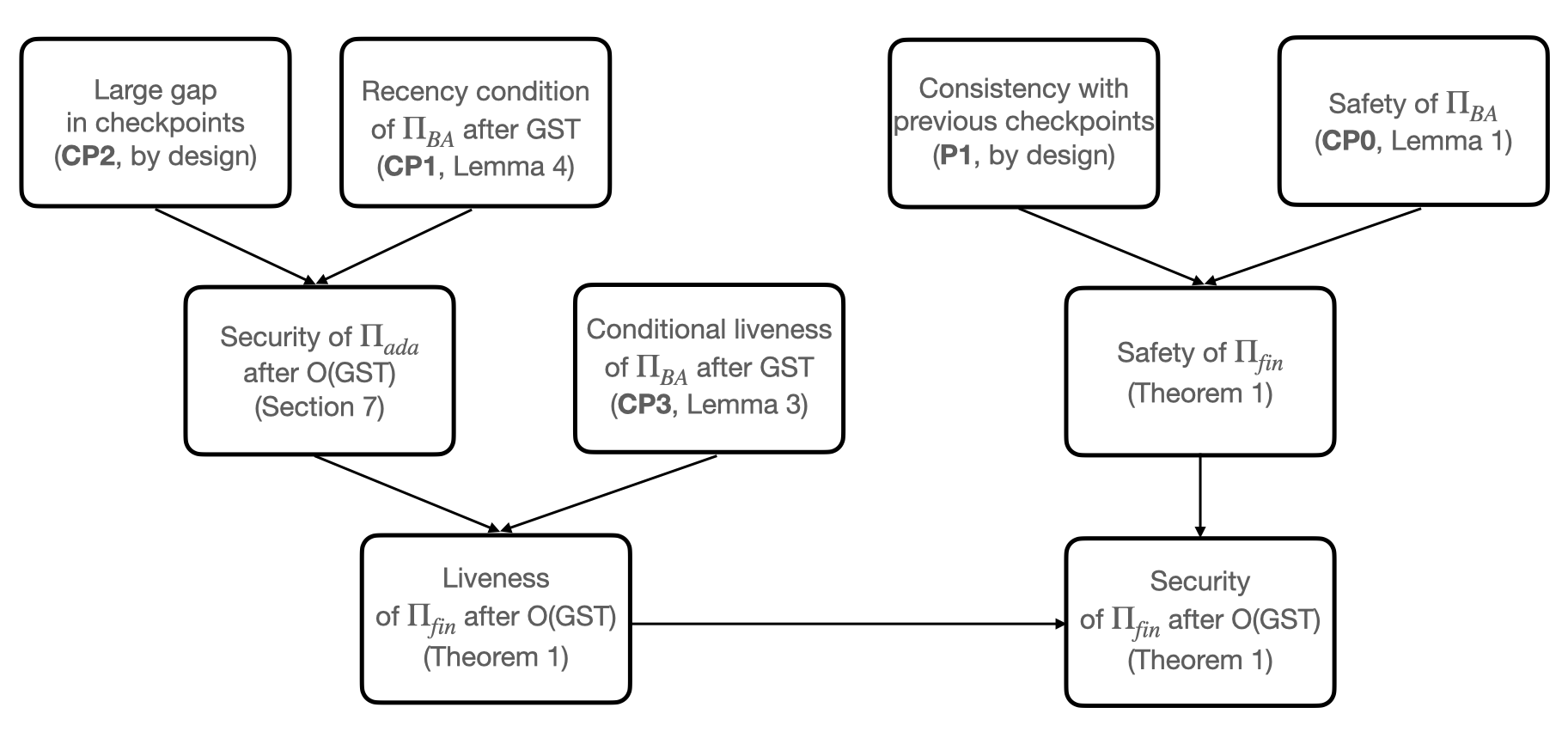}
    \caption{Flowchart for proving the security property of $\Pi_{\rm fin}$ under setting (${\rm M1, U1}$)}
    \label{fig:proof_map1}
\end{figure}
\begin{figure}[htbp]
    \centering
    \includegraphics[width=0.5\linewidth]{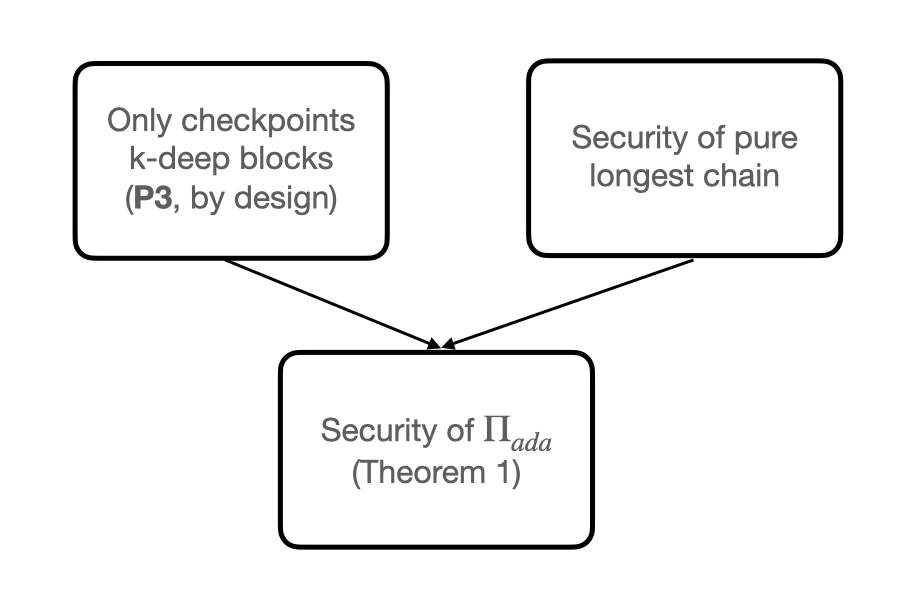}
    \caption{Flowchart for proving the security property of $\Pi_{\rm ada}$ under setting (${\rm M2, U2}$)}
    \label{fig:proof_map2}
\end{figure}
\begin{figure}[htbp]
    \centering
    \includegraphics[width=0.5\linewidth]{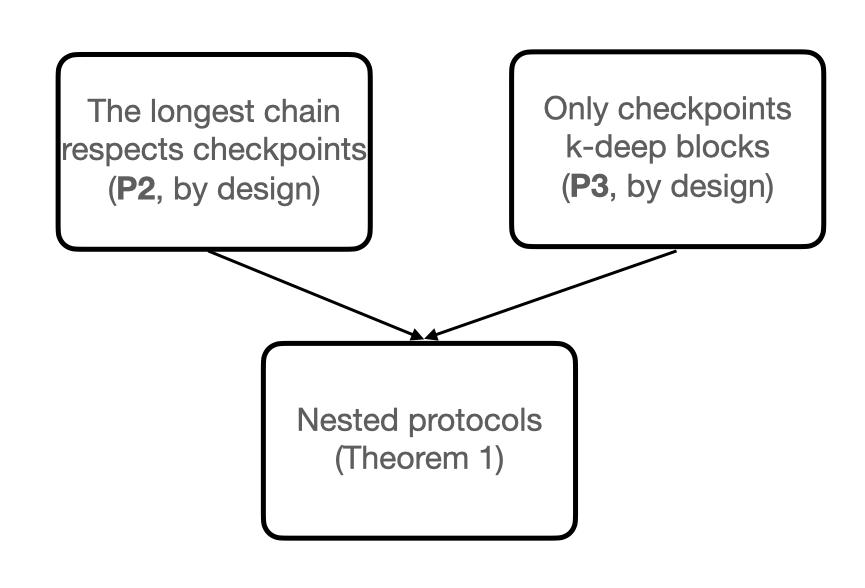}
    \caption{Flowchart for proving the nesting property of $\Pi_{\rm fin}$ and $\Pi_{\rm ada}$}
    \label{fig:proof_map3}
\end{figure}
We first highlight some aspects of Theorem \ref{thm:main} that are straightforward to deduce, assuming that $\Pi_{BA}$ does satisfy the aforementioned properties.

Firstly, the safety of $\Pi_{\rm fin}$ under setting (${\rm M1, U1}$) can be deduced from two facts: 1) ({\bf CP0}) safety of $\Pi_{BA}$ holds under setting (${\rm M1, U1}$), and 2) ({\bf P1}) the checkpointed longest chain rule respects all previous checkpoints (see Fig. \ref{fig:proof_map1}). Specially, {\bf CP0} of $\Pi_{BA}$ ensure that all checkpoints are uniquely decided by all honest nodes and they lie on a single chain by {\bf P1}, while the checkpointed longest chain protocol instructs honest nodes to always adopt the longest chain with the last checkpoint block. Therefore, an honest node will always keep every checkpoint block it has seen forever. 

Secondly, the security of $\Pi_{\rm ada}$ under setting (${\rm M2, U2}$) are deduced by two facts: 1) security of pure longest chain rule; 2) ({\bf P3}) the checkpointing protocol respects the $k$-deep rule (see Fig. \ref{fig:proof_map2}). {\bf P3} ensures that every checkpoint must be at least $k$-deep in the chain of some honest node when it is decided by $\Pi_{BA}$, while a block that is $k$-deep should remain in the longest chain of all honest nodes forever with high probability under synchrony \cite{pass2017analysis,ren2019analysis,dembo2020everything}. Therefore, every checkpoint will already be present in all honest nodes' chains when it appears. Thus, the chains held by the nodes of the protocol are indistinguishable from the case where they are following the pure longest chain protocol. Given that the confirmation rule is also the same for both protocols ($k$-deep rule), $\Pi_{\rm ada}$ inherits the safety and liveness properties of the pure longest chain protocol. 

Thirdly, the nesting property of the confirmed blocks also follow immediately by design, given that any checkpointed block must have been already $k$-deep in a node's chain ({\bf P3}) and the checkpointed longest chain is defined as the longest chain that contains the latest checkpoint ({\bf P2}), (see Fig. \ref{fig:proof_map3}).

It remains to prove the liveness guarantee of $\Pi_{\rm fin}$ under setting (${\rm M1, U1}$). In Appendix \iftoggle{arxiv}{\ref{sec:checkpointing_property}}{B (see full version \cite{sankagiri2020blockchain})}, we evaluate the safety and liveness of our checkpointing protocol $\Pi_{BA}$ -- a modified version of Algorand BA. We show that it satisfies the desired checkpointing properties listed in \S\ref{sec:intro}. A distinguishing property of $\Pi_{BA}$ is the recency condition ({\bf CP1}). This property states that $\Pi_{BA}$ only outputs checkpoints that are recently contained in some honest node's chain. Without this property, the adversary could make all honest nodes waste a long time mining on a chain that will never be checkpointed. 

In Appendix \iftoggle{arxiv}{\ref{sec:checkpointed_LC}}{C (see full version \cite{sankagiri2020blockchain})}, we prove an important result of $\Pi_{\rm ada}$ under partial synchronous model: the safety and liveness property will hold after $O(\rm{GST})$ time. This is essential to guarantee the liveness of $\Pi_{\rm fin}$ as a block can be only checkpointed when it lies in the $k$-common prefix of the checkpointed longest chains of all honest nodes. Finally, combining all these results (see Fig. \ref{fig:proof_map1}), we show that $\Pi_{\rm fin}$ is live after $O(\rm{GST})$ time with high probability, Thereby completing the proof of the main theorem.

\section{Conclusion}\label{sec:conclusion}
In this paper, we presented the design of a new finality gadget, to be used with Proof-of-Work blockchains. The proposed gadget provides each user a choice between an adaptivity guaranteeing confirmation rule and a finality guaranteeing one. This paper underscores the fact that it is possible to circumvent the impossibility result suggested by the CAP theorem by appropriately combining the longest-chain protocol with a committee-based BFT protocol. However, this paper is only a first step towards designing a viable protocol. Several interesting directions of research remain open, which we highlight below.

In our protocol, we used Algorand BA as our checkpointing protocol since it satisfies properties {\bf CP0-CP3}. Other natural candidates, such as PBFT \cite{castro1999practical}, Hotstuff \cite{yin2019hotstuff} and Streamlet \cite{chan2020streamlet} do not satisfy these properties. In particular, they do not satisfy the recency condition {\bf CP1} (we elaborate on this in Appendix \iftoggle{arxiv}{\ref{sec:discussion}}{D of the full version \cite{sankagiri2020blockchain}}). It would be interesting to see whether these conditions are necessary, or merely required for the proof. In the latter case, many more BFT protocols could be used for checkpointing.

We have shown that our protocol is essentially a finality gadget, just as Afgjort \cite{DinsdaleYoung2020AfgjortAP}, GRANDPA \cite{stewart2020grandpa} and Gasper \cite{buterin2020combining}. Through Figure \ref{fig:related}, it we have shown how some of these protocols could be tweaked to enhance their functionality to the level of the protocol we have designed. Formally analyzing the security of GRANDPA, and tweaking Gasper to make it secure (with a formal proof) are interesting open problems that could be tackled using the tools of this paper.

Finality gadgets/checkpointing could potentially offer many more properties. For example, they could protect against a dishonest mining majority, as shown in \cite{karakostas2020securing}. They could potentially also be used to have lower latency, and even responsive confirmation of blocks. Designing a protocol that achieves these properties in addition to the ones we show is an exciting design problem. Designing an incentive-compatible finality gadget also remains an open problem. Finally, a system implementation of such protocols could lead to newer considerations, such as communication complexity, latency, which would pave the way for future research.
\section{Acknowledgements}
This research is supported in part by a gift from IOHK Inc., an Army Research Office grant W911NF1810332 and by the National Science Foundation under grants CCF 17-05007 and CCF 19-00636. 

\bibliographystyle{acm}
\bibliography{references}

\appendix

\section*{Appendix}
\section{Algorand BA is a Checkpointing Protocol}
\label{sec:checkpointing}

We outline the full Algorand Byzantine Agreement (BA) protocol below for completeness. A minor modification (marked in red), adding validation, is the only addition to the original protocol. Honest checkpointers run a multi-iteration BA to commit checkpoints. The goal of the $i$-th iteration is to achieve consensus on the $i$-th checkpoint. Each iteration is divided into multiple periods and view changes will happen across periods.

All checkpointers start period $1$ of iteration $1$ at the same time (time $0$). Checkpointer $i$ starts period $1$ of iteration $n$ after it receives $2t + 1$ cert-votes for some value $v$ for the same period $p$ of iteration $n-1$ and waits for another fixed time $e$, and only if it has not yet started an iteration $n' > n$. Checkpointer $i$ starts period $p \geq 2$ of iteration $n$ after it receives $2t + 1$ next-votes for some value $v$ for period $p - 1$ of iteration $n$, and only if it has not yet started a period $p' > p$ for the same iteration $n$, or some iteration $n' > n$. For any iteration $n$, checkpointer $i$ sets its starting value for period $p \geq 2$, $st^p_i$, to $v$ (the value for which $2t+1$ next-votes were received and based on which the new period was started). For $p = 1$, $st^1_i = \perp$. The moment checkpointer $i$ starts period $p$ of iteration $n$, he finishes all previous periods and iterations and resets a local timer ${clock}_i$ to 0. At the beginning of every period, every honest checkpointer $i$ sets $v_i$ to be the checkpointed longest chain in its view at that time.

Each period has a unique leader known to all checkpointers. The leader is assigned in an i.i.d. fashion by the permitter oracle $O$. Each checkpointer $i$ keeps a timer ${clock}_i$ which it resets to 0 every time it starts a new period. As long as $i$ remains in the same period, ${clock}_i$ keeps counting. Recall that we assume the checkpointers’ individual timers have the same speed. In each period, an honest checkpointer executes the following instructions step by step.

\textbf{Step 1: [Value Proposal by leader]} The leader of the period does the following when ${clock}_i = 0$; the rest do nothing. If $(p = 1)$ OR $((p \geq 2)$ AND (the leader has received $2t + 1$ next-votes for $\perp$ for period $p - 1$ of iteration $n))$, then it proposes its value $v_i$.
Else if $((p \geq 2)$ AND (the leader has received $2t + 1$ next-votes for some value $v \neq \perp$ for period $p - 1$ of iteration $n))$, then the leader proposes $v$.

\textbf{Step 2: [The Filtering Step]}  Checkpointer $i$ does the following when ${clock}_i = 2\Delta$. If$(p = 1)$ OR $((p \geq 2)$ AND ($i$ has received $2t + 1$ next-votes for $\perp$ for period $p - 1$ of iteration $n))$, then $i$ soft-votes the value $v$ proposed by the leader of current period \textcolor{red}{if (it hears of it) AND (the value is VALID OR $i$ has received $2t+1$ next-votes for $v$ for period $p-1$)}. Else if $((p \geq 2)$ AND ($i$ has received $2t + 1$ next-votes for some value $v \neq \perp$ for period $p - 1$ of iteration $n))$, then $i$ soft-votes $v$.

\textbf{Step 3: [The Certifying Step]} Checkpointer $i$ does the following when ${clock}_i \in (2\Delta,4\Delta)$. If $i$ sees $2t + 1$ soft-votes for some value $v \neq \perp$, then $i$ cert-votes $v$.

\textbf{Step 4: [The Period’s First Finishing Step]}  Checkpointer $i$ does the following when ${clock}_i = 4\Delta$. If $i$ has certified some value $v$ for period $p$, he next-votes $v$. Else if $($ $(p \geq 2)$ AND ($i$ has seen $2t+1$ next-votes for $\perp$ for period $p-1$ of iteration $n$)$)$, he next-votes $\perp$. Else he next-votes his starting value $st^p_i$.

\textbf{Step 5: [The Period’s Second Finishing Step]} Checkpointer $i$ does the following when ${clock}_i \in (4\Delta,\infty)$ until he is able to finish period $p$. If $i$ sees $2t + 1$ soft-votes for some value $v\neq \perp$ for period $p$, then $i$ next-votes $v$. If $( (p \geq 2)$ AND $(i$ sees $2t + 1$ next-votes for $\perp$ for period $p-1$) AND ($i$ has not certified in period $p$)$)$, then $i$ next-votes $\perp$.

Halting condition: Checkpointer $i$ HALTS current iteration if he sees $2t+1$ cert-votes for some value $v$ for the same period $p$, and sets $v$ to be his output. Those cert-votes form a certificate for $v$. The block that is exactly $k$-deep in $v$ is chosen as the checkpoint in current iteration.

A proposed value $v$ (with block $B$ being exactly $k$-deep in it) from period $p$ of iteration $n$ is VALID for checkpointer $i$ (in the same period and iteration) if:
\begin{itemize}
    \item Value $v$ is proposed by the leader of that period;
    \item Block $B$ is a descendant of all previously checkpointed blocks with smaller iteration number;
    \item Block $B$ is contained in the checkpointed longest chain that the checkpointer $i$ holds when entering period $p$.
\end{itemize}

The only modification to the protocol is in Step 2, in the first condition, where the notion of validity is introduced. This is the only place where new proposals are considered. The validity notion helps transform the Algorand BA protocol into the multi-iteration checkpointing protocol that we desire. In Appendix \iftoggle{arxiv}{\ref{sec:checkpointing_property}}{B (see full version \cite{sankagiri2020blockchain})}, we show that this protocol indeed satisfies properties {\bf CP0-CP3}, as mentioned in \S\ref{sec:protocol}.
\iftoggle{arxiv}
{
\section{Checkpointing Properties}
\label{sec:checkpointing_property}

In this section, we prove several important properties of our modified Algorand BA protocol. We assume throughout that the number of adversarial checkpointers is strictly less than one-third the total number of checkpointers; this is necessary to achieve any security results for a partially synchronous consensus protocol such as Algorand BA.
\begin{itemize}
    \item It is safe even under {asynchronous network conditions}({\bf CP0}, proved in Lemma \ref{lem:ba_safe}). Moreover, it remains deadlock free, in the sense that it can always proceed as long as messages will be delivered eventually (Lemma \ref{lem:deadlock_free}).
    \item It is live after $\textsf{GST}$ if the local checkpointed longest chains of all honest nodes satisfy common prefix property ({\bf CP3}, proved in Lemma \ref{lem:cond_live}).
    \item It only outputs recently inputted values after $\textsf{GST}$ ({\bf CP1}, proved in Lemma \ref{lem:recent}).
\end{itemize}

We first define some useful notations according to the BA protocol described in Appendix \ref{sec:checkpointing}.

\begin{definition}[Potential starting value for period $p$]
A value $v$ that has been next-voted by $t+1$ honest nodes for period $p-1$.
\end{definition}

\begin{definition}[Committed value for period $p$]
A value $v$ that has been cert-voted by $2t+1$ nodes for period $p$.
\end{definition}

\begin{definition}[Potentially committed value for period $p$]
A value $v$ that has been cert-voted by $t+1$ honest nodes for period $p$.
\end{definition}

Although we slightly altered Algorand BA protocol (which is highlighted in red in Appendix \ref{sec:checkpointing}), we note that our modification does not break the safety of the protocol or cause any deadlock in Lemma \ref{lem:ba_safe} and Lemma \ref{lem:deadlock_free}. At a high level, the validity check only causes less soft-votes from honest nodes, which is indistinguishable with the case where the leader is malicious and no value receives  at least $2t+1$ soft-votes in some period. Therefore, the safety and deadlock-free property remain.

\begin{lemma}[Asynchronous Safety, {\bf CP0}]
\label{lem:ba_safe}
Even when the network is partitioned, the protocol ensures safety of the system so that no two honest nodes will finish one iteration of the protocol with different outputs.
\end{lemma}

\begin{proof}
The following properties hold even during a network partition.
\begin{itemize}
    \item By quorum intersection, as each honest node only soft-votes one value, then at most one value is committed or potentially committed for each period $p$ in one iteration.
    \item If a value $v$ is potentially committed for period $p$, then only $v$ can receive $2t+1$ next-votes for period $p$. Thus, the unique potential starting value for period $p+1$ is $v$.
    \item If a period $p$ has a unique potential starting value $v \neq \perp$, then only $v$ can be committed for period $p$. Moreover, honest nodes will only next-vote $v$ for period $p$, so the unique potential starting value for period $p+1$ is also $v$. Inductively, any future periods $p'>p$ can only have $v$ as a potential starting value. Thus, once a value is potentially committed, it becomes the unique value that can be committed or potentially committed for any future period, and no two honest nodes will finish this iteration of the protocol with different outputs.
\end{itemize}

\end{proof}

\begin{lemma}[Asynchronous Deadlock-freedom]
\label{lem:deadlock_free}
As long as messages will be delivered eventually, an honest node can always leave period $p$, either by entering a higher period or meeting the halting condition for the current iteration.
\end{lemma}

\begin{proof}
We first prove that there can never exist $2t+1$ next-votes for two different non-$\perp$ values from the same period $p$ by induction.

Start with $p = 1$. Note that every honest node sets $st^1_i = \perp$ and at most one value (say $v$) could receive more than $2t+1$ soft-votes. Therefore only value $v$ and $\perp$ could potentially receive more than $2t+1$ next-votes in period $1$. Note that it is possible that both $v$ and $\perp$ receive more than $2t+1$ next-votes: all the honest nodes could next-vote for $\perp$ in Step 4 and then next-vote for $v$ in Step 5 after seeing the $2t+1$ soft-votes for $v$.

Assume that the claim holds for period $p-1$ ($p \geq 2$): there exist at most two values each of which has $2t + 1$ next-votes for period $p-1$, and
one of them is necessarily $\perp$. Then there are three possible cases:

\begin{itemize}
    \item Case 1: Only $\perp$ has $2t + 1$ next-votes for period $p-1$. This is the same as the induction basis $p=1$.
    \item Case 2: Only $v$ has $2t+1$ next-votes for period $p-1$. Then in period $p$, every honest node will have starting value $st^p_i = v$, so only $v$ could receive soft-vote and next-vote.
    \item Case 3: Both $\perp$ and $v$ have $2t + 1$ next-votes for period $p-1$. Suppose a new value $v' \neq v$ receives $2t+1$ next-votes for period $p$, then $v'$ must have $2t+1$ soft-votes for period $p$ and at least $t+1$ soft-votes are from honest nodes, which means that at least $t+1$ honest nodes have received $2t + 1$ next-votes for $\perp$ for period $p - 1$ when they enter Step 2. Hence when they enter Step 4, they will next-vote $v'$ or $\perp$ instead of $v$. And they will not next-vote $v$ in Step 5 either. Then $v$ can have at most $t$ next-votes from honest nodes and at most $2t$ next-votes from all nodes for period $p$.
\end{itemize}
 
Therefore, we proved the claim for period $p$. Then we prove the lemma. Note that a node will stuck in period $p$ if and only if no value receives at least $2t+1$ next-votes for period $p$. Then by the previous claim we proved, there are only two possible cases:
\begin{itemize}
    \item Case 1: $\perp$ has $2t + 1$ next-votes for period $p-1$. If no value has at least $2t+1$ soft-votes for period $p$, then every honest node will just next-vote $\perp$ in Step 4 or Step 5 (since $2t+1$ next-votes on $\perp$ for period $p-1$ will eventually reach every honest node). Otherwise if some value $v$ has at least $2t+1$ soft-votes for period $p$, then every honest node will next-vote $v$.
    \item Case 2: Only some value $v$ has $2t+1$ next-votes for period $p-1$. Then in period $p$, every honest node will have starting value $st^p_i = v$, so only $v$ could receive soft-vote and next-vote. And every honest node will next-vote $v$ in Step 4.
\end{itemize}
\end{proof}

The original Algorand BA guarantees fast recovery from network partition: after $\textsf{GST}$, an agreement is reached in $O(\Delta)$ time. However, our protocol does not have such a strong liveness property as honest nodes may disagree on what is the longest chain for some time right after $\textsf{GST}$. Here in Lemma \ref{lem:cond_live}, we prove that our protocol is live and each iteration takes $O(\Delta)$ time to finish once the common prefix property of all honest nodes' longest chains is recovered, which may take $O(\textsf{GST})$ time as we will see in \S\ref{sec:checkpointed_LC}.  

\begin{lemma}[Conditional Liveness After $\textsf{GST}$, {\bf CP3}]
\label{lem:cond_live}
Suppose after time $T$, all honest nodes always have the same values, i.e., their local checkpointed longest chains satisfy common prefix property, then after $\max\{{\textsf{GST}},T\}$, each iteration of Algorand BA will finish within $O(\Delta)$ time. Moreover, all honest nodes finish the same iteration within time $\Delta$ apart. 
\end{lemma}

\begin{proof}

Let $N$ be the largest iteration number such that some honest node has committed some value $v$ before $\textsf{GST}$. Then all honest nodes will receive the certificate for $v$ within time $\Delta$ after $\textsf{GST}$ and they will all enter iteration $N+1$ after some time. If such $N$ does not exist, i.e., no value has been committed by any honest node before $\textsf{GST}$, then we set $N=0$.

Further, let $p$ be the highest period that some honest node is working on in iteration $N+1$ right before $\textsf{GST}$ is. Then after time $\Delta$, all honest nodes will also start period $p$ as they receive $2t + 1$ next-votes for period $p - 1$. Soon after, some value $v$ (which may be $\perp$) will receive at least $2t+1$ for period p by Lemma \ref{lem:deadlock_free}, and all honest node will all start period $p + 1$ within time $\Delta$ apart.

Once all honest nodes start the same period $p$ within time $\Delta$ apart, the network partition seems to have never happened.
Without partition, if all honest nodes have the same value, i.e., after $\max\{\textsf{GST},T\}$, the following facts hold:
\begin{itemize}
    \item By Lemma \ref{lem:deadlock_free}, there exist at most two values each of which has $2t+1$ next-votes for every period, and one of them is necessarily $\perp$ if there exist exactly two.
    \item If a period $p$ is reached and the leader is honest, then all honest nodes finish the current iteration in period $p$ by their own time $6\Delta$, with the same output $v \neq \perp$. Moreover, all honest nodes finish within time $\Delta$ apart. Indeed, if there exists a potentially committed value $v$ for period $p-1$, then $v$ will be the unique potential starting value for period $p$. In period $p$, the honest leader proposes $v$ in Step 1 and all honest nodes soft-vote $v$ in Step 2. In Step 3, by their own time $4\Delta$, all honest nodes have cert-voted $v$. Thus all of them finish the current iteration by their own time $6\Delta$, with output $v$. Otherwise if there is no potentially committed value in period $p-1$, then the leader of period $p$ may propose its private input $v'$ or a value $v \neq \perp$ for which it has seen $2t+1$ next-votes from period $p-1$. In the first case, all honest nodes will soft-vote $v'$ in Step 2 as they all see $2t+1$ next-votes for $\perp$ from period $p-1$; in the second case, all honest nodes will soft-vote v in Step 2 either by condition 1 or condition 2. In both cases, all honest nodes will cert-vote the same value in Step 3 by their own time $4\Delta$ and finish the current iteration with that value by their own time $6\Delta$.
    \item If a period $p$ is reached and there is no committed value for period $p$ (which only happens if the leader is malicious), then all honest nodes move to period $p+1$ by their own time $8\Delta$, and they move within time $\Delta$ apart. Note that the honest nodes may start period $p$ not at the same time but within time $\Delta$ apart. Indeed, if no honest node has cert-voted in Step 3, then all honest nodes next-vote a common starting value $v$ in Step 4 if no one has seen $2t+1$ next-votes for $\perp$ for period $p-1$, or else all honest nodes next-vote $\perp$ in Step 4 or Step 5. In both cases, they all see these next-votes and move to period $p+1$ by their own time $8\Delta$. If at least one honest node has cert-voted a value $v$ in Step 3, then $v$ must have received $2t + 1$ soft-votes and no other value could have these many soft-votes by quorum intersection. Thus only $v$ could have received cert-votes. Since those honest nodes have helped rebroadcasting the soft-votes for $v$ by their own time $4\Delta$, all honest nodes see $2t+ 1$ soft-votes for $v$ by their own time $6\Delta$. Thus all honest nodes next-vote $v$ either in Step 4 or Step 5 by their own time $6\Delta$, and all see $2t+1$ next-votes for the same value by their own time $8\Delta$. Note that some honest nodes may have next-voted for $\perp$ in Step 4 as well, thus there may also exist $2t + 1$ next-votes for $\perp$.
    \item Combining the above facts together, since the leader in each period is honest with probability $> 2/3$, one iteration takes in expectation at most 1.5 periods and at most $10\Delta$ time. Moreover, all honest nodes finish within time $\Delta$ apart.
\end{itemize}

\end{proof}

The liveness property showed in Lemma \ref{lem:cond_live} conditions on the assumption that all honest nodes will have the same values after some time $T$. However, this assumption may not hold right after $\textsf{GST}$ as the adversary could have a bank of private blocks that can be used to break the common prefix property. During this stage, Algorand BA may still output some values in favor of the adversary (eg., the adversary sends its favorite chain to all honest nodes). But if the output value is too old in the sense that no honest node has that checkpointed block in its local longest chain for a long time, then all honest blocks mined during this time interval could be potentially wasted as they are not extending the last checkpoint. However, we show that Algorand BA with our cautious modification won't let this happen in the following lemma. 

\begin{lemma}[Recency Condition After $\textsf{GST}$, {\bf CP1}]
\label{lem:recent}
After $\textsf{GST}$, if one honest node commits a value $v$ in one iteration at time $t$, then $v$ must have been the input value held by at least one honest node at time no earlier than $t - 8D\Delta$, where $D \sim \rm{Geo}(p)$ with $p>2/3$ being the fraction of honest nodes.
\end{lemma}

\begin{proof}
For a fixed iteration, We define $V_p$ to be the set of potential starting values for period $p$. Then $V_1 = \{ \perp \}$ and also by Lemma \ref{lem:deadlock_free}, we have $1 \leq |V_p| \leq 2$, and $\perp \in V_p$ if $|V_p| = 2$ for all $p$. Note that $V_p = \{v, \perp\}$ is an bad set as $v$ may be an outdated value. Then we show that the state won't remain in $\{v, \perp\}$ for the same $v$ for a long time. 

Suppose the leader of period $p$ is honest and no honest node has value $v$ for period $p$, then we have the following two cases.
\begin{itemize}
    \item Case 1: if the leader of period $p$ has received $2t+1$ next-votes for $\perp$ for period $p-1$ at the beginning of Step 1, then it will propose $v' \neq v$. Further no honest node will soft-vote or next-vote $v$. Therefore, the state of period $p+1$ will be $\{\perp\}$, $\{v'\}$, or $\{v', \perp\}$.
    \item Case 2: if the leader of period $p$ has received $2t+1$ next-votes only for $v$ for period $p-1$ at the beginning of Step 1, then it will propose $v$. Further every honest node will soft-vote $v$ in Step 2 either by condition 1 or condition 2. Therefore, $v$ will be committed in period $p$.
\end{itemize}

Therefore, after $\textsf{GST}$, the state won't remain in $\{v, \perp\}$ for the same $v$ from period $p$ to period $p+1$, if the leader of period $p$ is honest and no honest node has value $v$ for period $p$. By Lemma \ref{lem:cond_live}, we know that after $\textsf{GST}$ each period (even with a malicious leader) will finish in $8\Delta$ time. Combining these facts together, we can conclude the lemma.

\end{proof}

We say the $i$-th checkpoints committed at some time $t$ by an honest node after $\textsf{GST}$ is $D_i$-{\it recent} if it has been held at depth exactly $k$ in the chain of an honest node at some time $t' \geq t - D_i$. Then by Lemma $\ref{lem:recent}$, we have $D_i$ is stochastically dominated by $8D\Delta$, where $D \sim \rm{Geo}(p)$ with $p>2/3$. Let us define event $E$ to be the event that all checkpoints committed after $\textsf{GST}$ and before $T_{\rm max}$ are $d$-recent, where $d$ is called the recency parameter of the protocol. Then we have the following corollary directly from Lemma $\ref{lem:recent}$.

\begin{corollary}
\label{coro:uniform}
There exists $d = O(\sqrt{\kappa} \Delta)$ such that $\mathbb{P}(E) \geq 1-e^{-\Omega(\sqrt \kappa)}$, where $\kappa$ is the security parameter. 
\end{corollary}
\begin{proof}
Recall that there can be at most $T_{\rm max}/e$ checkpoints by {\bf CP2}. By the union bound, we have
\begin{equation*}
    \mathbb{P}(E^c) \leq \sum_i \mathbb{P}(D_i > d) \leq \frac{T_{\rm max}}{e} \mathbb{P}( {\rm Geo}(p) > \frac{d}{8\Delta}) = \frac{T_{\rm max}}{e} (1-p)^{\frac{d}{8\Delta}} = e^{-\Omega(\sqrt \kappa)}.
\end{equation*}
\end{proof}
Thus we conclude that the modified Algorand BA, given in Appendix \ref{sec:checkpointing}, provides the desired checkpointing properties {\bf CP0-CP3}.

\section{Checkpointed Longest Chain Properties}\label{sec:checkpointed_LC}
\subsection{Comparison With Longest Chain Protocol}
The security properties of the longest chain protocol has been intensely studied in recent years. The strong security properties have been demonstrated in increasing sophistication (both network models as well as tightness of security threshold):  the pioneering work of \cite{garay2015bitcoin} on the round by round network model has been extended to discrete and finely partitioned network model in \cite{pass2017analysis} and to a continuous time network model in \cite{ren2019analysis}. The tightness of the security threshold has been improved to the best possible in \cite{dembo2020everything}. Despite this wealth of technical background, the  {\em checkpointed} longest chain protocol has seemingly subtle differences with the vanilla longest chain protocol, but impact the analysis significantly. We discuss two of these issues next. 

\subsubsection{Synchrony versus Partial Synchrony}. The  security analyses of the longest chain protocol require that synchronous network conditions hold from the beginning of the protocol \cite{garay2015bitcoin,pass2017analysis,dembo2020everything}. 
In the partially synchronous model, where messages may be arbitrarily delayed until some unknown time $\textsf{GST}$, the protocol's security breaks down. More precisely, even after (a bounded time beyond) $\textsf{GST}$, safety and liveness may fail. We describe a possible attack scenario in Figure \ref{fig:partial_synchrony_attack}. 
\begin{figure}[htbp]
    \centering
    \includegraphics[width=\linewidth]{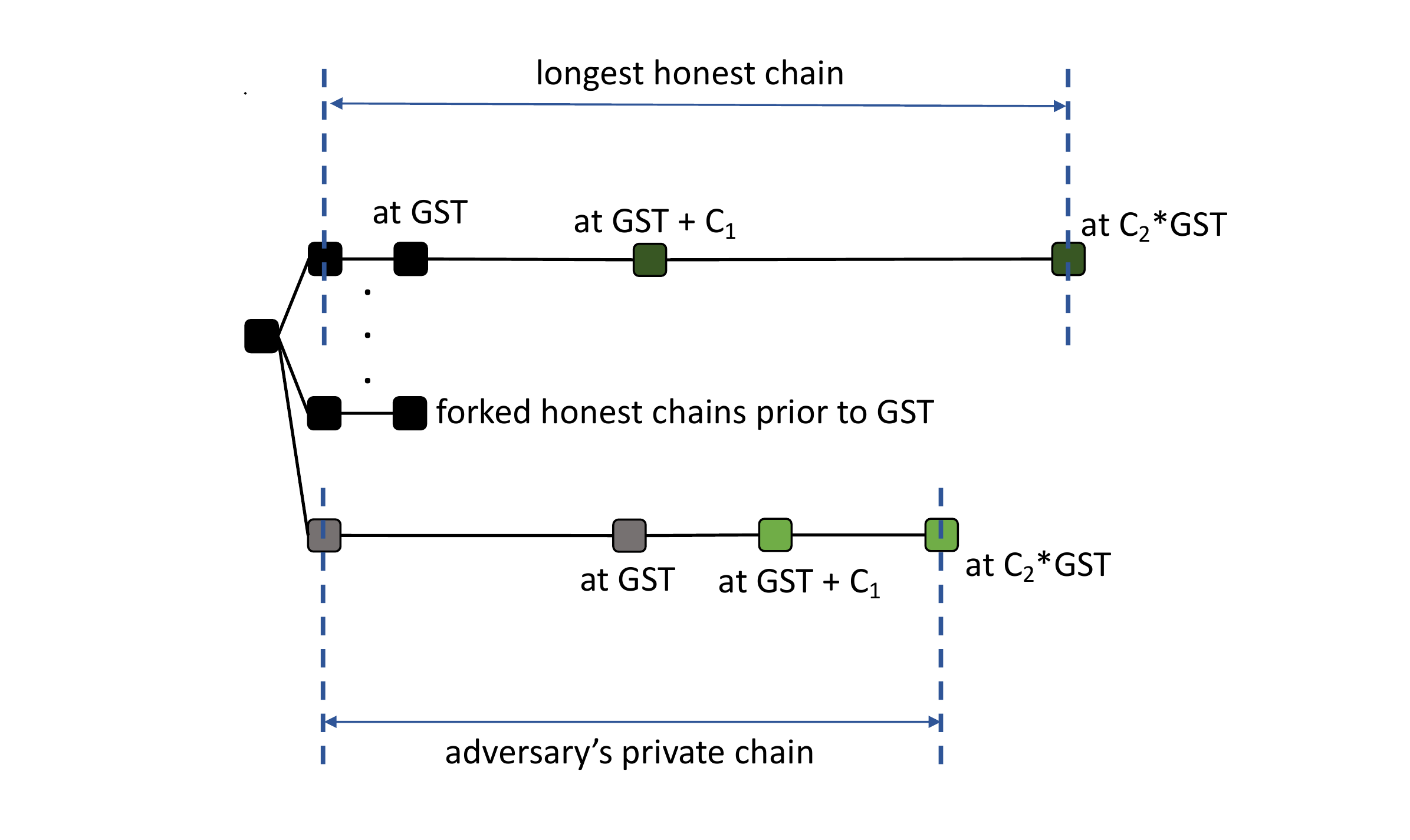}
    \caption{An illustration of the private attack by an adversary in a partially synchronous system. The private chain is longer than the longest honest chain for a short while after $\textsf{GST}$, but not indefinitely after.}
    \label{fig:partial_synchrony_attack}
\end{figure}
The aforementioned works do not provide any security guarantees in the partially synchronous model.  
However, the method of analysis suggests that it is possible to provide some security guarantees \textit{after $O(\textsf{GST})$ time}. This should hold for any adversarial power $\beta$ in which the protocol is secure in the synchronous regime. Consider the feasibility of a private attack, in which the adversary must mine a longer chain than that mined by the honest nodes. The adversary has some initial advantage immediately after \textsf{GST} by means of some private blocks (see Figure \ref{fig:partial_synchrony_attack}), but it will eventually lose out to the honest chain as the latter's growth rate is strictly larger. The analysis of \cite{dembo2020everything} suggests that the private attack is the worst-case attack. To the best of our knowledge, only \cite{neu2020ebb} analyzes a longest chain protocol (the sleepy consensus protocol of \cite{pass2017sleepy}) in the partially synchronous setting, and they show desirable properties hold after $O(\textsf{GST})$ time, with the constant being approximately $1/(1-2\beta)$.

\subsubsection{Non-Monotonic Chain Growth}. The checkpointed longest chain protocol can behave very differently compared to the classical longest chain protocol when checkpoints appear. In the classical version, the chain length of each honest node grows monotonically. In contrast, in our protocol, the chain held by an honest node may reduce in length once it hears of a new checkpoint. This could happen if its current chain does not contain the new checkpoint. This has important ramifications for security. A key property used in analyzing the classical protocol is that (after $\textsf{GST}$), two honest blocks that appear more than $\Delta$ time apart must be at different heights (see e.g., \cite{ren2019analysis}). This property needn't hold for the checkpointed longest chain protocol; in fact, it can be violated if the adversary chooses to do so. In effect, this causes a bleeding (or wastage) of honest mining power. An instance of this is illustrated in Figure \ref{fig:same_height_loners}. 
\begin{figure}[htbp]
    \centering
    \includegraphics[width=\linewidth]{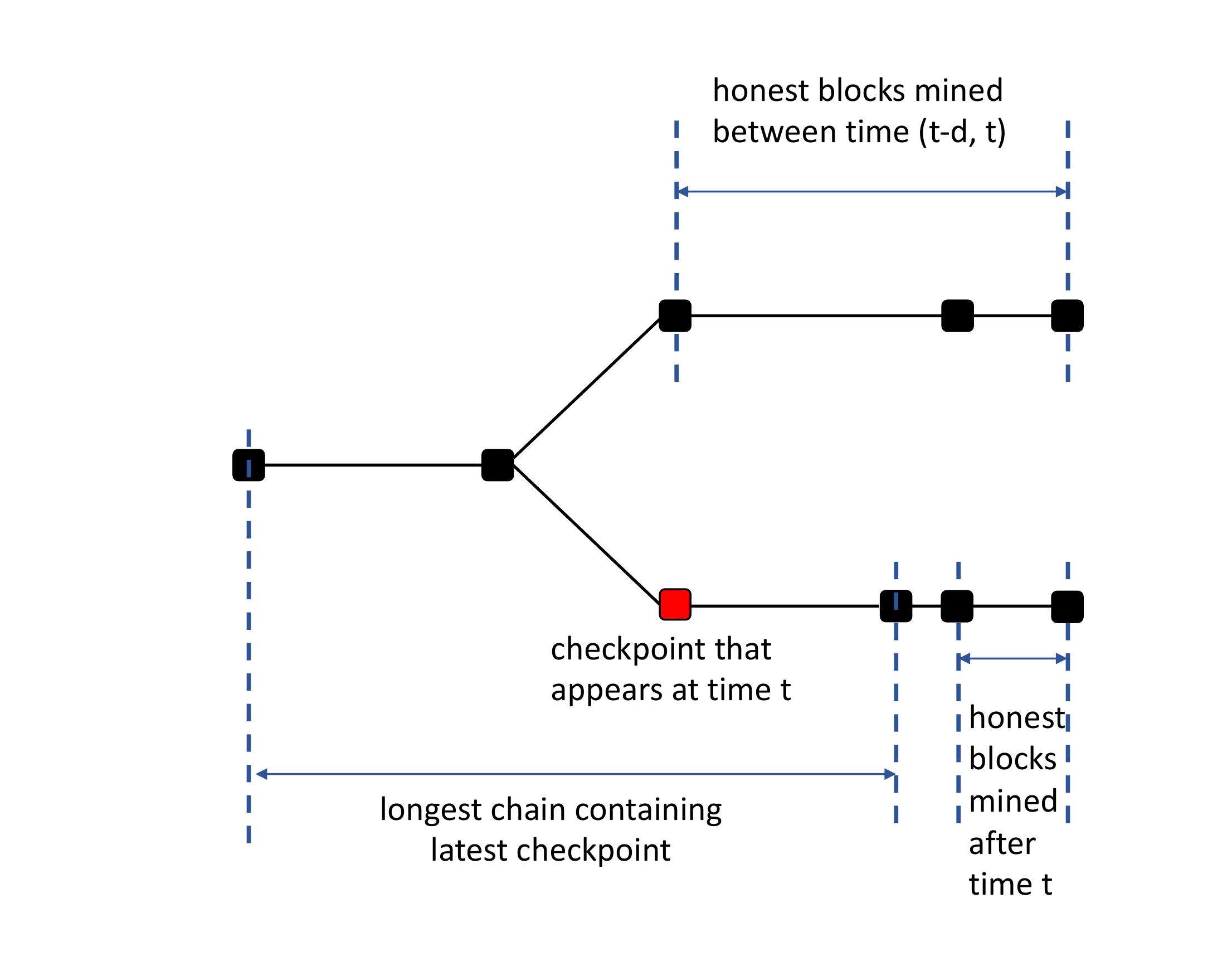}
    \caption{An illustration of a reduction in chain length due to the appearance of a checkpoint. Consequently, two honest blocks may be at the same height, even if they are mined more than $\Delta$ time apart.}
    \label{fig:same_height_loners}
\end{figure}
The analysis of \cite{neu2020ebb} uses many properties of longest chain protocols `off the shelf' from \cite{pass2017sleepy}. However, for reasons highlighted above, we must establish new properties that hold deterministically for our protocol, irrespective of the adversary's actions. These properties are weaker analogs of those satisfied by the longest chain protocol. This is to be expected, since the checkpointed longest chain protocol may behave just as the longest chain protocol irrespective of $\textsf{GST}$, if the adversary so wishes.
\subsection{Definitions}
We first define the notion of common prefix and chain quality, which are now standard in the literature (see, e.g., \cite{garay2015bitcoin}).
\begin{definition}[$k$-Common Prefix]\label{def:common_prefix}
{For a chain $\mathcal{C}$, let $\mathcal{C} \lfloor k$ denote the new chai nobtained by dropping the last $k$ blocks from $\mathcal{C}$. We say that the $k$-common prefix property holds if for any two chains $\mathcal{C}_1$, $\mathcal{C}_2$ held by honest nodes at times $t_1 \leq t_2$ respectively, it holds that $\mathcal{C}_1 \lfloor k \preceq \mathcal{C}_2$.}
\end{definition}

\begin{definition}[$(k, s)$-Chain Quality]
We say that the $(k, s)$-chain quality property holds if for all honestly held chains $\mathcal{C}$, every portion of $\mathcal{C}$ with $k$ consecutive blocks mined after time $s$ contains at least one honest block.
\end{definition}

Comparing to \cite{garay2015bitcoin}, our definition of $k$-common prefix is identical to theirs. Our definition of $(k, s)$-Chain Quality with $s = 0$ is identical to their definition of chain quality with parameter $\mu = 1/k$. We introduce the extra parameter $s$ to emphasize that the chain quality property holds only for the portion of the chain consisting of blocks mined after time $s = O(\textsf{GST})$. In the partially synchronous model, we cannot hope to do better. In \cite{garay2015bitcoin}, the authors show that common prefix is identical to safety (also called persistence) and chain quality is identical to liveness for the $k$-deep confirmation rule. Thus, to prove security of $\Pi_{\rm ada}$ after $O(\textsf{GST})$, it suffices to show that $k$-common prefix and $(k,s)$-chain quality hold after time $O(\textsf{GST})$, for $s = O(\textsf{GST})$.

\subsection{Proof Sketch}
In this section, we analyze the security of the protocol by discretizing time into slots. Note that this is purely a proof strategy; we do not change any assumptions about the behavior of the nodes, nor of the message timings. Our proof is structured similar to that in \cite{garay2015bitcoin}. It has three major parts that we describe below.
\begin{itemize}
    \item First, we define the notion of a typical execution (Definition \ref{def:typical}). A typical execution is one in which the number of mined honest and adversarial blocks over a sufficiently large period do not deviate from their mean by a large fraction. We introduce the notion of a \textit{convergence opportunity} here, similar to \cite{pass2017analysis}. We show that a typical execution occurs with high probability (Corollary \ref{coro:typical}).
    \item Next, we show necessary conditions for common-prefix violation and chain quality violation in terms of the mining process. These conditions are analogous to the condition in the classical longest chain protocol, which are roughly as follows: the number of adversarial blocks over an interval is more than the number of honest blocks over the same interval.
    \item Third, we show that these necessary conditions won't happen after $O(\textsf{GST} + \kappa)$ time in a typical execution (Lemma \ref{lem:prob}). Hence, with high probability, the common prefix property and the chain quality property of the checkpointed longest chain will hold after $O(\textsf{GST} + \kappa)$ time (Corollary \ref{coro:common}).
\end{itemize}
The first part is similar to existing works. Our main novelty lies is the second and third parts. As hinted above, we can at best hope to obtain weaker necessary conditions for common prefix violation and chain quality violation than those for the classical protocol. At the same time, a condition that is too weak may not be rare in a typical execution. We show that the necessary conditions are very similar for both $k$-common prefix violation and $(k, s)$-chain quality violation. Moreover, we show that these conditions are impossible in a typical execution after $O(\textsf{GST} + \kappa)$ time. Our analysis relies heavily on the properties of the checkpointing protocol proved in \S\ref{sec:checkpointing_property}. For clarity, we further split this part of the proof into four steps.
\begin{itemize}
    \item We reduce all common-prefix violation instances to a special case where two diverging chains are held by honest nodes at the same slot.
    \item We establish a lower bound on the lengths of honestly held chains in terms of convergence opportunities in an interval, with some slack to account for the possibility of bleeding.
    \item We establish the desired necessary condition for $k$-common prefix violation by analyzing the common portion and the forked portion of the two diverging honestly held chains.
    \item We establish the desired necessary condition for $(k, s)$-chain quality violation by comparing the chain growth for honest chains in an interval with the number of adversarial blocks in the same interval.
\end{itemize}
In this section, whenever we refer to a block $B$ being $k$-deep in a chain $\mathcal{C}$, we mean there are $k$ or more blocks that are descendants of $B$ in $\mathcal{C}$.

\subsection{Typical Execution}
From \S\ref{sec:protocol}, we know that the block mining process is a Poisson process with rate $\lambda$. This can be further split into an adversarial block mining process of rate $\beta\lambda$ and an independent honest block mining process of rate $(1-\beta)\lambda$, with $\beta < 1/2$. Therefore the number of honest/adversarial blocks mined in any interval is a Poisson random variable. By the Chernoff bound, we know that a Poisson random variable is tightly concentrated around its mean. However, we would like to obtain such a result for \textit{all intervals} of length bigger than $O(\kappa)$. Hence, for the sake of applying the union bound, we discretize time into slots: the $i\textsuperscript{th}$ slot is the time interval $((i-1) \Delta,i \Delta]$. The total number of slots in the entire execution of the system is $t_{\rm max} = \lceil T_{\rm max}/\Delta \rceil = O({\rm poly}(\kappa))$. We abuse notation to let $\textsf{GST}$ denote the first slot that begins after the global synchronization time. All our analysis henceforth will be in the form of slots, but the protocol remains a continuous time one as described in \S\ref{sec:protocol}.

We define the following random variables for slot $i$.
\begin{itemize}
    \item $Y_i = 1$ if there is exactly one block that is mined by honest nodes in slot $i$ and there is no block that is mined by honest nodes in slot $i-1$ and $i+1$, otherwise $0$. We follow the notation in \cite{pass2017analysis} to call such slots as a {\it convergence opportunity} (such a block is also called {\it loner} in \cite{ren2019analysis}).
    \item $Z_i$ is the number of adversarial blocks plus the number of honest blocks that are not a convergence opportunity, in slot $i$.
\end{itemize}
Let \[{\bar y} \triangleq \mathbb{E}[Y_i] = (1-\beta)\lambda\Delta\ e^{-3(1-\beta)\lambda\Delta}.\] Further, $\mathbb{E}[Y_i + Z_i] = \lambda\Delta$, and hence \[{\bar z} \triangleq \mathbb{E}[Z_i]= \lambda \Delta - \bar y.\]

For any $t_1 \leq t_2$, we define $Y[t_1, t_2] = \sum_{i=t_1+1}^{t_2} Y_i$ and $Z[t_1, t_2] = \sum_{i=t_1+1}^{t_2} Z_i$. With some abuse of notation, we let $Y[S]$ denote $\sum_{i\in S}Y_i$ for any set of slots $S \subset \mathbb{N}$. As we are now in discrete time, the notation $[t_1, t_2]$ denotes the set of slots $\{t_1, \ldots, t_2\}$.
We have that \[\mathbb{E}Y[t_1, t_2] = {\bar y}(t_2 - t_1), \quad \mathbb{E}Z[t_1, t_2] = {\bar z}(t_2 - t_1).\] 
We note that for any $\beta < 1/2$ and any $\Delta$, there exists $\varepsilon > 0$ and $\lambda > 0$ s.t. \[{\bar y}(1 - 3\varepsilon) > {\bar z}.\] In particular, this implies ${\bar y} > {\bar z}$, and $(1 - \varepsilon){\bar y} > {\bar z} + \varepsilon {\bar y}$. We shall operate under this regime. We also assume that $\lambda \Delta < 1$. This condition implies $\bar y = \Omega(\lambda \Delta)$, where the constants hidden by $\Omega(\lambda \Delta)$ do not depend on $\beta, \lambda, \Delta$.

With these notations in place, we can now define the notion of a {\it typical execution}.
\begin{definition}[Typical execution]
\label{def:typical}
An execution is $(\varepsilon, \tau)$-typical 
for a $\beta$-corrupt system, 
for $\varepsilon \in (0,1)$
and $\tau \in \mathbb{N}$,
if the following events hold
{for any $t_1, t_2$ s.t. $t_2 - t_1 \geq \tau$}:
\begin{align*}
    |Y[t_1, t_2]- \mathbb{E}Y[t_1, t_2]|&\le \varepsilon\mathbb{E}Y[t_1, t_2]; \\
    |Z[t_1, t_2]-\mathbb{E}Z[t_1, t_2]| &\le \varepsilon\mathbb{E}Y[t_1, t_2];\\
    |Y[t_1, t_2] + Z[t_1, t_2]-(\mathbb{E}Y[t_1, t_2] + \mathbb{E}Z[t_1, t_2])| &\le \varepsilon(\mathbb{E}Y[t_1, t_2] + \mathbb{E}Z[t_1, t_2]).
\end{align*}
\end{definition}

\begin{proposition}
\label{prop:ineq} For any $t_1 \leq t_2$, the following events hold with probability at least {$1-e^{-\Omega(\varepsilon^2 \lambda \Delta |t_2 - t_1|)}$}: 
\begin{align*}
    |Y[t_1, t_2]- \mathbb{E}Y[t_1, t_2]|&\le \varepsilon\mathbb{E}Y[t_1, t_2]; \\
    |Z[t_1, t_2]-\mathbb{E}Z[t_1, t_2]| &\le \varepsilon\mathbb{E}Y[t_1, t_2];\\
    |(Y[t_1, t_2] + Z[t_1, t_2])-(\mathbb{E}Y[t_1, t_2] + \mathbb{E}Z[t_1, t_2])| &\le \varepsilon(\mathbb{E}Y[t_1, t_2] + \mathbb{E}Z[t_1, t_2]).
\end{align*}
\end{proposition}

\begin{proof}
The Chernoff bound for a binomial or a Poisson random variable $X$ gives us
\[\mathbb{P}(|X - \mathbb{E}X| \geq \delta \mathbb{E}X) \leq 2e^{-\delta^2\mathbb{E}X/3} \]

We first prove the concentration bound on $Y[t_1, t_2]$. We cannot directly apply standard concentration inequalities on $Y[t_1, t_2]$ since $Y_i$'s are not independent random variables. However, for any $i$, $Y_i$ and $Y_{i+3}$ are independent. Each $Y_i$ is Bernoulli({$\bar y$}). Hence, $Y[t_1, t_2]$ can be written as the sum of three binomial random variables.
\[Y[t_1, t_2] = \sum_i Y_{t_1 + 3i} + \sum_i Y_{t_1 + 3i + 1} + \sum_i Y_{t_1 + 3i + 2}.\]
Using the Chernoff bound and the union bound, we get that each of these three binomial random variables are within a factor of $\varepsilon/3$ from their mean, except with probability with probability $e^{-\Omega(\varepsilon^2 \lambda \Delta (t_2 - t_1))}$. When this happens, the desired bound on $Y[t_1, t_2]$ follows. 

Next, we note that $Y[t_1, t_2]+Z[t_1, t_2]$ is a Poisson random variable. Using the Chernoff bound for Poisson random variables, and the fact that $\mathbb{E}(Y[t_1, t_2]+Z[t_1, t_2]) = \lambda \Delta (t_2 - t_1)$, we obtain that $Y[t_1, t_2]+Z[t_1, t_2]$ is within a factor of $\varepsilon$ from its mean except with probability $e^{-\Omega(\varepsilon^2 \lambda \Delta (t_2 - t_1))}$. 

It remains to bound $Z[t_1, t_2]$. For this, we need a slightly tighter bound on $Y[t_1, t_2]$ and $Y[t_1, t_2] + Z[t_1, t_2]$ than the one in the statement of the proposition; we replace $\varepsilon$ by $\varepsilon/3$. Both these bounds also hold except with probability $e^{-\Omega(\varepsilon^2 \lambda \Delta (t_2 - t_1))}$. From these bounds, and the relation ${\bar y} > {\bar z}$, we can the desired bound on $Z[t_1, t_2]$.
\end{proof}

There are at most $t_{\rm max}^2$ choices of $t_1, t_2$. The following corollary follows from Proposition \ref{prop:ineq} and the union bound.

\begin{corollary}\label{coro:typical}
The protocol executes with ($\varepsilon,\tau$)-\textit{typical execution} with probability at least $1-e^{-\Omega(\varepsilon^2\kappa)}$.  
\end{corollary}
\begin{proof}
From Proposition \ref{prop:ineq} and the union bound, we get that the probability of a typical execution is at least $1- t_{\rm max}^2 e^{-\Omega(\varepsilon^2 \lambda \Delta \tau)}$. We can choose $\tau$ large enough such that $\lambda\Delta\tau=\Theta(\kappa)$. We know that $t_{\rm max} = {\rm poly}(\kappa)$. From this, the stated result follows.
\end{proof}


\subsection{Necessary Conditions for Common Prefix/Chain Quality Violation}
We begin by noting some basic properties of the checkpointing protocol, and also introduce some new notation that is useful for this subsection. Our first result is a chain growth lemma (Lemma \ref{lem:chain_growth}). The chain growth lemma highlights the extent of bleeding that can take place in the checkpointing protocol. Next, we focus on the common prefix property. We show that without loss of generality, we can assume that common-prefix violation occurs between two chains at the same slot (Lemma \ref{lem:cp_reduction}). Following this, we show that in the forked portion, the number of adversarial blocks over the forked portion must be at least half the total number of blocks. This gives us the desired necessary condition comparing the number of convergence opportunities and adversarial blocks over a sufficiently large period (Lemma \ref{lem:necessary}). Finally, we analyze the chain quality property in Lemma \ref{lem:chain_quality_necessary}.

\subsubsection{Basic properties of checkpoints}
Let us recall some properties of the checkpointing algorithm, re-stated in terms of slots instead of discrete time.
\begin{enumerate}
    \item We say that a checkpoint appears in slot $t$ if the first time an honest node marks the checkpoint is in slot $t$. If a checkpoint appears at slot $t$, all honest nodes mark the checkpoint by the beginning of slot $t+2$. 
    \item Let $d' \triangleq \lceil d/\Delta \rceil$. Any checkpoint that appears at slot $t > \textsf{GST} + d'$ was at least $k$-deep in a chain held by an honest node in some slot $t' \in \{t - d', \ldots t\}$.
    \item Let $e' = \lfloor e/\Delta\rfloor$. If two consecutive checkpoints appear at slots $t, t'$, then $t' \geq t + e'$.
\end{enumerate}

Let $\textsf{GST} + d' < t^*_1 \leq t^*_2 \leq \ldots$ be the sequence of slots at which checkpoints appear after slot $\textsf{GST} + d'$. Let the corresponding checkpoint blocks be $B_1^*, B_2^*, \ldots$. Let $t^*_0 \triangleq \textsf{GST}$. 
Define $S_l \triangleq \{t^*_l + 2, \ldots, t^*_{l+1} - d' - 2\}$ for all $l \geq 0$ and let $S \triangleq \cup_{l \geq 0} S_l$. We refer to each $S_l$ as an inter-checkpoint interval.

\subsubsection{Chain growth results}

\begin{lemma}[Chain Length Lower Bound]\label{lem:chain_length_lowet_bound}
Let $t_1, t_2$ be two slots such that $t_1 \in S$ and and $t_2 \geq t_1 + 2$. Let $\mathcal{C}$ be a chain held by some honest node in slot $t_1$. Then all honest nodes will hold a chain of length at least $|\mathcal{C}|$ in slot $t_2$.
\end{lemma}
\begin{proof}
Let $h$ be an honest node that holds chain $\mathcal{C}$ in slot $t_1$, and let $h'$ be an arbitrary honest node (possibly $h$) that holds a chain $\mathcal{C}'$ in slot $t_2$. Since $t_1 \in S$, we know that $t \in S_l$ for some $l$. Therefore $B^*_l$ is the last checkpoint that is marked in $\mathcal{C}$. Clearly, all honest nodes will hear of $\mathcal{C}$ by the beginning of slot $t_2$, and would also have marked $B^*_l$ as a checkpoint. We aim to show that $|\mathcal{C}'| \geq |\mathcal{C}|$. There are now two cases to consider.
%
\begin{itemize}
    \item Node $h'$ has not marked any new checkpoint at the time it holds $\mathcal{C}'$. In this case, it will hold a chain at least as long as $\mathcal{C}$, as it is a chain that it has heard of that contains all the checkpoints that it has marked.
    \item Node $h'$ has marked at least one new checkpoint at the time it holds $\mathcal{C}'$. In particular, it has marked $B^*_{l+1}$. Note that $B^*_{l+1}$ must be exactly $k$-deep in some honestly held chain $\mathcal{C}^*$ at some slot in the interval $[t^*_{l+1} - d', t^*_{l+1}]$ in which $B^*_{l+1}$ was not marked as a checkpoint. By the first case, $\mathcal{C}^*$ must be at least as long as $\mathcal{C}$, since $t_1 \leq t^*_{l+1} - d' - 2$. Further, all honest nodes must hold a chain at least as long as $\mathcal{C}^*$ at the time of marking $B^*_{l+1}$. Arguing inductively, we see that all honest nodes that have marked at least one new checkpoint compared to $\mathcal{C}$ also hold a chain at least as long as $\mathcal{C}$.
%
\end{itemize}
\end{proof}

\begin{lemma}[Chain Growth]\label{lem:chain_growth}
Consider a chain $\mathcal{C}$ held by an honest node in a slot $t \geq \textsf{GST} + 1$. Suppose $\mathcal{C}$ contains an honest block $B$ mined in slot $s$ at height $\ell$. Then $|\mathcal{C}| \geq \ell + Y[S \cap [s+2, t-2]]$.
\end{lemma}
\begin{proof}
From Lemma \ref{lem:chain_length_lowet_bound}, we can deduce three useful facts. Firstly, the first converge opportunity block in $S \cap [s+2, t-2]$ will be built at a height $\geq \ell + 1$, since all honest nodes by that slot would have adopted a chain at least as long as $\ell$.  Secondly, we see that all convergence opportunities in $S$ are mined at distinct heights. This is because all convergence opportunities are spaced at least $2$ slots apart. Thirdly, all honest nodes will adopt a chain that is at least as long as the chain mined by the last converge opportunity in $S \cap [s+2, t-2]$. Together, these two statements prove the desired statement when $s \in S$.

To prove the same result when $s \not \in S$, we consider the following two cases:

\begin{itemize}
    \item Suppose $s \notin S$, and $s \geq \textsf{GST}$. This implies $s \in \{t_l^* - d' - 1, \ldots t_l^* + 1\}$ for some $l \geq 1$. then in slot $t_l^*+2$ all honest nodes will see both block $B$ and checkpoint $B_l$. Moreover, block $B$ and checkpoint $B_l$ are on the same chain $\mathcal{C^*}$, hence all honest nodes should hold a chain of length at least $\ell$ in slot $t_l^* + 2$, simply because $\mathcal{C^*}$ contains $B$. The rest of the proof is identical to the case $s \in S$.
    \item Suppose $s \notin S$ and $s<\textsf{GST}$. Let $\mathcal{B}$ be the set of checkpoints with appearance time $<= \textsf{GST}$, i.e, for a block $B^* \in \mathcal{B}$, at least one honest node marks $B^*$ before $\textsf{GST}$. Then in slot $\textsf{GST} + 1$, all honest nodes will see a chain $\mathcal{C^*}$ that contains block $B$ and all checkpoints in $\mathcal{B}$. Hence all honest nodes should hold a chain of length at least $\ell$ in slot $\textsf{GST}+1$, simply because $\mathcal{C^*}$ contains $B$. The rest of the proof is identical to the case $s \in S$.
\end{itemize}

\end{proof}

\subsubsection{Analysis of Common Prefix Property}
Let us first note a simple fact about the checkpointed longest chain rule. In what follows, we say a chain $\mathcal{C}$ was held by an honest node in slot $t$ if it was held by the node at any time in the interval $((t-1)\Delta, t\Delta]$. 

\begin{proposition}\label{lem:cp_persistence}
If a block $B$ is $k$-deep in a chain $\mathcal{C}$ held by an honest node $h$ in slot $t$, then it either remains $k$-deep in $\mathcal{C}'$ or is not present in $\mathcal{C'}$, where $\mathcal{C'}$ is an chain held by $h$ in slot $t'$.
\end{proposition}
The above statement is trivially true for the longest chain rule, since the chain length held by an honest node is monotonically increasing. It is not so immediate for the checkpointed longest chain rule, since the length of the chain held by a particular honest node may reduce upon adopting a new checkpoint block. A crucial point in the proof of this statement is the property that a checkpoint is always $k$-deep in all honest nodes' chains. The statement is illustrated in Figure \ref{fig:k_deep_persistence}, and is proven below.

\begin{proof}
As remarked above, the property holds trivially as long as the honest node $h$ adopts successive chains of increasing length. Consider the case when at some slot $t' \geq t$, $h$ adopts a new checkpoint block and truncates its chain in the process. Suppose block $B$ is present in both $\mathcal{C}$ and $\mathcal{C'}$. Then it must be an ancestor of the new checkpoint block. Since the new checkpoint block is at least $k$-deep in the new chain, so is block $B$.
\end{proof}
\begin{figure}[h]
\begin{centering}
\includegraphics[scale=0.4]{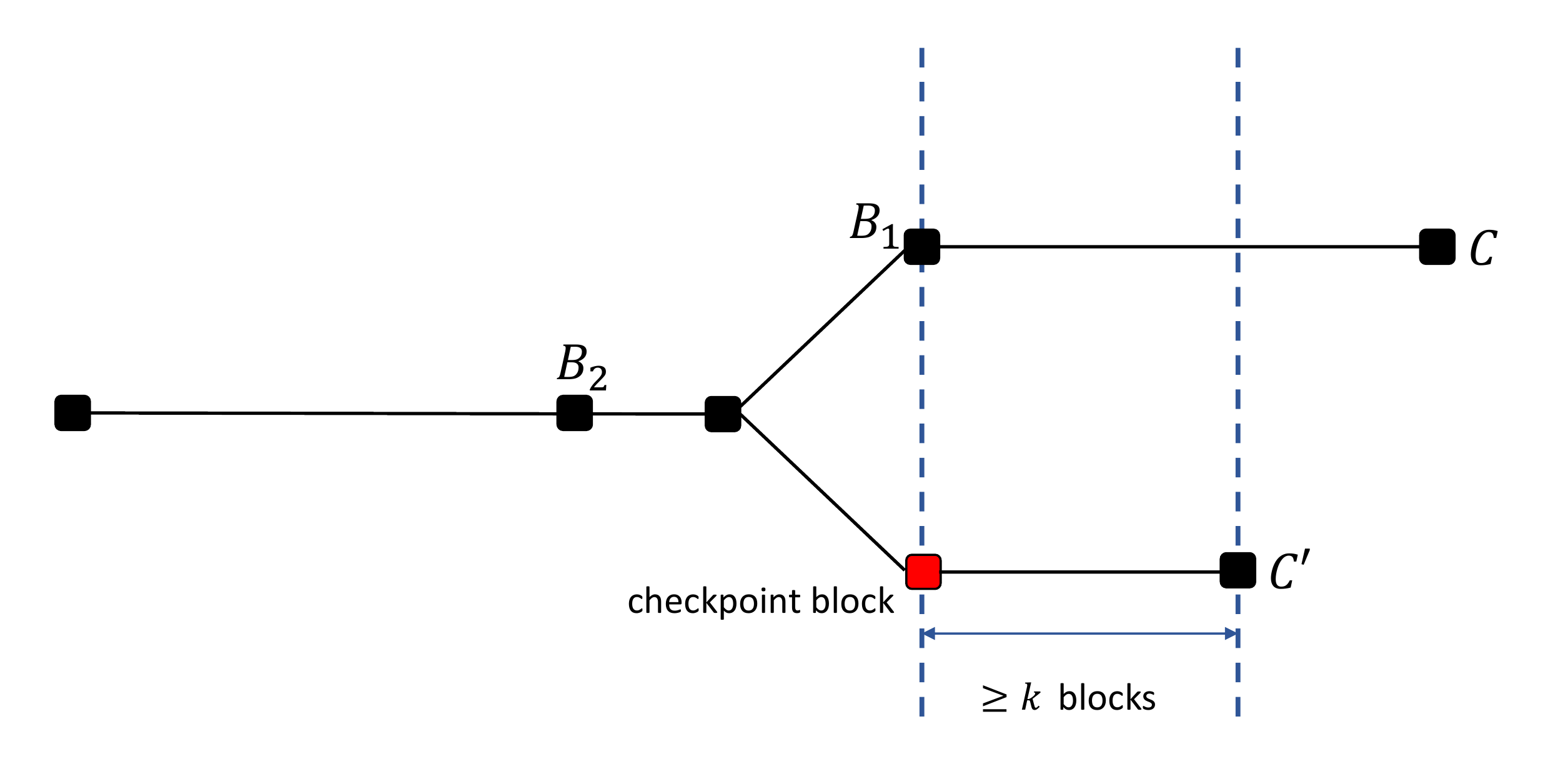}
\caption{Illustration of Lemma \ref{lem:cp_persistence}. Both $B_1$, $B_2$ are $k$-deep in $\mathcal{C}$. $B_1$ gets ousted in $\mathcal{C'}$, while $B_2$ remains $k$-deep.}
\end{centering}
\label{fig:k_deep_persistence}
\end{figure}

\begin{lemma}[Common Prefix Reduction]\label{lem:cp_reduction}
Suppose the common prefix property is violated after slot $t_0$, i.e., there exist slots $t_1$ and $t_2$: $t_0 \leq t_1 \leq t_2$, honest nodes $h_1, h_2$ (possibly the same), and chains $\mathcal{C}_1$, $\mathcal{C}_2$, s.t. $\mathcal{C}_1$ is held by $h_1$ at slot $t_1$, $\mathcal{C}_2$ is held by $h_2$ at slot $t_2$, but $\mathcal{C}_1 \lfloor k \npreceq \mathcal{C}_2$. {Then the following event happens at some slot $t \in [t_1, t_2]$: there exists chain $\mathcal{C}$ and $\mathcal{C}'$ such that $\mathcal{C}' \lfloor k \npreceq \mathcal{C}$ while both $\mathcal{C}'$ and $\mathcal{C}$ are held by honest nodes in slot $t$.}
\end{lemma}

We refer to this events as \textit{common prefix violation at slot $t$}, and the pair of chains $\mathcal{C}'$ and $\mathcal{C}$ as a \textit{certificate of the violation}. Lemma \ref{lem:cp_reduction}, stated for the checkpointed longest chain rule, is also true for the longest chain rule. The proofs are identical as well, given that the statement of Lemma \ref{lem:cp_persistence} holds for both protocols. We provide the proof for completeness.

\begin{proof}
Let $t \geq t_1$ be the earliest slot at which there exists some chain $\mathcal{C}$ held by an honest node $h$ (possibly $h_1$) s.t. $\mathcal{C}_1 \lfloor k \npreceq \mathcal{C}$. Let $B$ be the block that is exactly $k$-deep in $\mathcal{C}_1$; clearly, it is not present in $\mathcal{C}$. We now claim that there must be a chain $\mathcal{C}'$ held by node $h_1$ at some time in slot $t$ that contains $B$. Note that $\mathcal{C}$ and $\mathcal{C}'$ satisfy the relation $\mathcal{C}' \lfloor k \npreceq \mathcal{C}$, thus proving the desired statement. It remains to show the claim.

Suppose $h_1$ never holds a chain containing $B$ in slot $t$. Then it must have adopted a chain that does not contain $B$ in an earlier slot. However, we have defined $t$ to be the earliest slot such that some chain held by an honest node does not contains $B$. Thus, we have shown the claim to be true. A technical point to note is that if an honest node adopts a new chain $\mathcal{C}_{\text{new}}$ at time $s$ in slot $t$, it must have held the older chain $\mathcal{C}_{\text{old}}$ at some time ($s-\epsilon$) \emph{in slot $t$ itself}.
%
\end{proof}
\begin{figure}[h]
\begin{centering}
\includegraphics[scale=0.4]{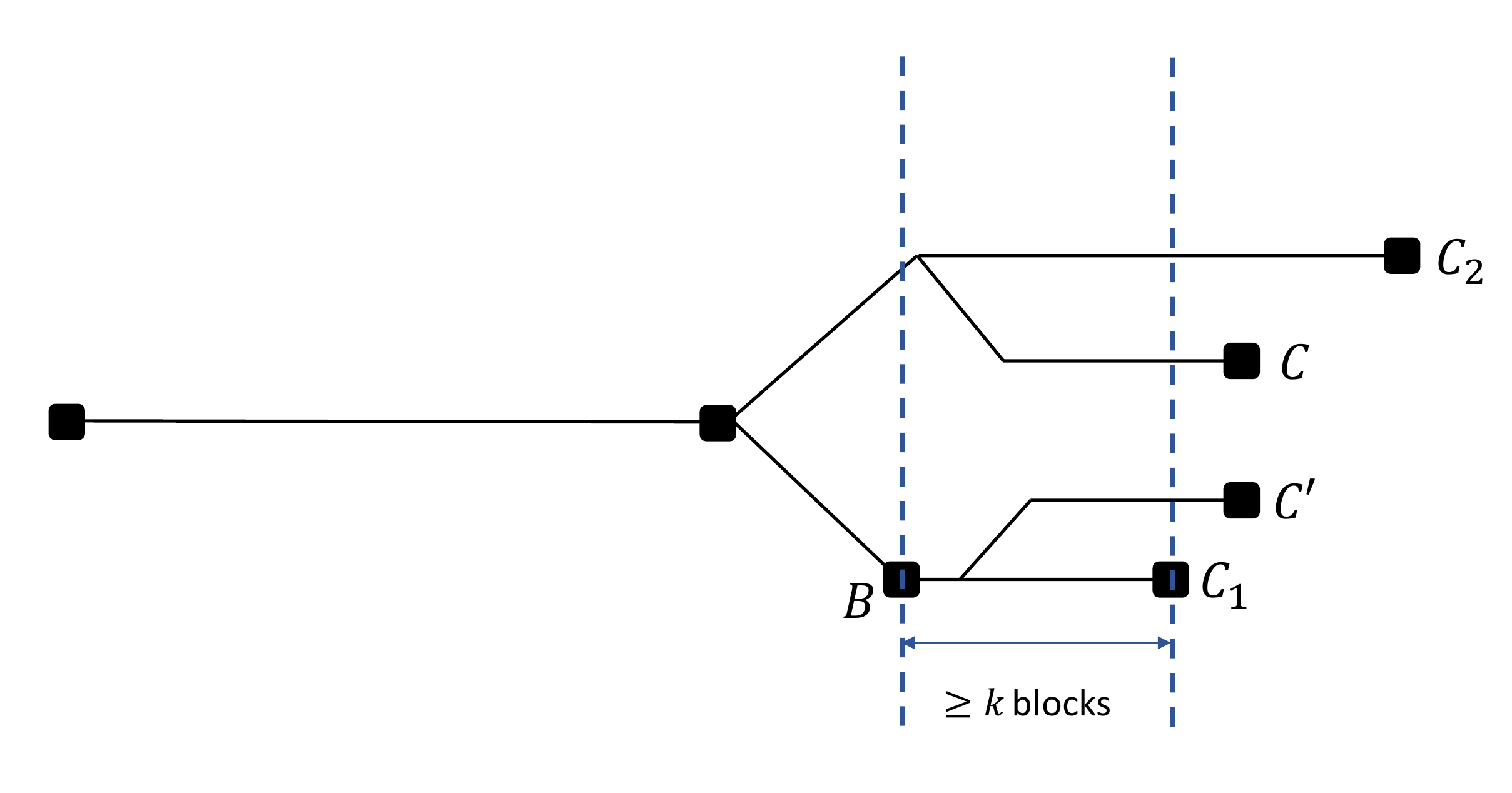}
\caption{Illustration of Lemma \ref{lem:cp_reduction}. Block $B$ is $k$-deep in $\mathcal{C}_1$. It remains so in $\mathcal{C}'$. However, it is absent from $\mathcal{C}$. In this case, $\mathcal{C}'$ and $\mathcal{C}$ form a certificate of $k$-common prefix violation.} 
\end{centering}
\label{fig:cp_reduction}
\end{figure}

\begin{lemma}[Necessary condition for common prefix violation]
\label{lem:necessary}
Suppose $k$-common prefix violation occurs at slot $t$. Then there exist $t^* < t$ such that $Z[t^*,t] + Y[t^*,t] \geq k$ and $Z[t^*, t] \geq Y[S \cap [t^*+2, t-2]] - 1$. 
\end{lemma}

\begin{proof}
The first condition $Z[t^*,t] + Y[t^*,t] \geq k$ is trivial, as a certificate of $k$-common prefix violation contains at least $k$ blocks. Then we focus on the derivation of the second condition.
From Lemma \ref{lem:cp_reduction}, we know that there exists two honestly held chains $\mathcal{C}, \mathcal{C}'$ such that $\mathcal{C}' \lfloor k \npreceq \mathcal{C}$. Let $\ell$ be the height of the shorter of the two chains. Consider the last block on the common prefix of $\mathcal{C}, \mathcal{C}'$ that is a convergence opportunity. Denote it by $\overline{B}$, let its height be $\overline{\ell}$ and let $\overline{t}$ be the slot at which it was created. {If there is no convergence opportunity in the common prefix, let $\overline{B}$ be the genesis block and $\overline{t} = 0$}. From Lemma \ref{lem:chain_growth}, we know that $\ell \geq \ell^* + Y[S \cap [t^*+2, t-2]]$, which implies $\ell - \ell^* \geq Y[S \cap [t^*+2, t-2]]$. We now show that $Z[t^*, t] \geq \ell - \ell^* - 1$.

Firstly, we note that any adversarial blocks in $\mathcal{C}, \mathcal{C}'$ that are descendants of $\overline{B}$ must be mined in the interval $[t^*, t]$. Let $\Tilde{B}$ denote the last block on the common prefix of $\mathcal{C}, \mathcal{C}'$, let its height be $\Tilde{\ell}$. The portion of the common prefix from block $\overline{B}$ (excluded) to $\Tilde{B}$ (included) is composed entirely of adversarial blocks. This accounts for $\Tilde{\ell} - \overline{\ell}$ adversarial blocks.

We now consider two possibilities regarding the number of checkpoints on the chains $\mathcal{C}, \mathcal{C}'$. In the first case, both chains have the same number of checkpoints marked. In this case, all checkpoints must be on the common prefix of $\mathcal{C}, \mathcal{C}'$. In the second case, one of the chains has one extra checkpoint. This checkpoint must have appeared in slot $t$ or slot $t-1$, for otherwise, it would be adopted by the other chain too. In both cases, let $\hat{B}$ denote the last \textit{common} checkpoint of the two chains. We now analyze the forked portion of the chains in both cases.

For the first case, we show that all convergence opportunity blocks that are descendants of $\hat{B}$ must be at distinct heights. Note that every chain ending at each such convergence opportunity blocks contain all checkpoints that appear until slot $t$. Thus, at the time when any honest node would have received such a chain (say $\Tilde{\mathcal{C}}$), it would adopt $\Tilde{\mathcal{C}}$ or a chain of at least $|\Tilde{\mathcal{C}}|$. This is because $\Tilde{\mathcal{C}}$ would contain all the checkpoints that this node would have marked until then. Since the miner of every convergence opportunity hears of broadcasts from all previous convergence opportunities before mining, the claimed result follows. For the second case, a slightly modified statement holds: all convergence opportunities that are descendants of $\hat{B}$ mined in slots $\leq t-2$ must be at distinct heights. This can be proven by arguing in the same manner as above. Since the latter is a weaker claim than the former, we use the latter claim in the general case.

Note that all blocks in the forked portion are descendants of $\Tilde{B}$, and are therefore descendants of $\hat{B}$. Therefore, all convergence opportunity blocks in the forked portion mined in slots $\leq t-2$ must be at distinct heights. Since they can be on at most one of the chains, they must be paired with either an adversarial block or a convergence opportunity from slot $t-1$ or $t$. There can be at most one convergence opportunity from slots $\{t-1, t\}$. From this, we conclude that the number of adversarial blocks in the forked portion is at least $\ell - \Tilde{\ell} - 1$.

In summary, we can say that the number of distinct adversarial blocks in chains $\mathcal{C}, \mathcal{C}'$ combined that are descendants of block $B^*$ is at least $\ell - \Tilde{\ell} - 1 + \Tilde{\ell} - \overline{\ell}$ $=$ $\ell - \overline{\ell} - 1$. All these blocks are mined in slots in the range $[t^*, t]$. Thus, $Z[t^*, t] \geq \ell - \ell^* - 1$. Combined with the inequality $\ell - \ell^* \geq Y[S \cap [t^*+2, t-2]]$, we get the inequality: $Z[t^*, t] \geq Y[S \cap [t^*+2, t-2]] - 1$.
\end{proof}

\subsubsection{Analysis of Chain Quality Property}

With some abuse of notation, we use the term $(k, s)$-chain quality to refer to a slot $s$ instead of some time $s$.

\begin{proposition}[Chain Quality Reduction]\label{prop:chain_quality_reduction}
Suppose at some slot $t$, $(k, s)$-chain quality is violated, i.e., there exists an honestly held chain $\mathcal{C}$ in slot $t$ that contains $k$-consecutive blocks mined after slot $s$ and there is no honest block amongst them. Then at some slot $t' \in (s, t]$, there exists a chain $\mathcal{C}'$ held by an honest node such that the last $k$ blocks of $\mathcal{C}'$ are all adversarial blocks.
\end{proposition}
\begin{proof}
Suppose, at some slot $t$, $(k, s)$-chain quality is violated. Clearly, $t > s$, else it is not possible for $\mathcal{C}$ to contain blocks mined at slots $> s$. There are two cases to consider. First, we consider the case where there are no more honest blocks in $\mathcal{C}$ after the string of $k$ consecutive adversarial blocks. Then the claim trivially follows with $t' = t$, $\mathcal{C}' = \mathcal{C}$. Next, consider the case where there exists an honest block in $\mathcal{C}$ after the string of $k$ consecutive adversarial blocks. Let the block be $B$ and let its mining time be $t'$. Since these $k$ blocks are all mined after slot $t_0$, so is $B$ (i.e., $t' > s$). Then the portion of $\mathcal{C}$ up to $B$ (but not including it) was held by the honest miner of block $B$. Call this sub-chain $\mathcal{C}'$. From the condition on $\mathcal{C}$, it is immediate that the last $k$ blocks of $\mathcal{C}'$ are adversarial. Thus we have our claim.
\end{proof}

If there exists an honestly held chain $\mathcal{C}$ at slot $t$ such that the last $k$ blocks of $\mathcal{C}$ are all adversarial blocks mined after slot $s$, we say that \emph{$(k, s)$-chain quality violation occurs at slot $t$}. From Proposition \ref{prop:chain_quality_reduction}, we assume without loss of generality that all violations of $(k, s)$-chain quality are such that $(k, s)$-chain quality violation occurs at slot $t$, for some $t > s$.

\begin{lemma}[Necessary Condition for Chain Quality Violation]\label{lem:chain_quality_necessary}
Suppose, for any slot $t \geq \textsf{GST} + 1$, \emph{$(k,s)$-chain quality is violated at slot $t$}. Then there $\exists$ $t^* < t$ s.t. $Z[t^*, t] \geq k$ and $Z[t^*, t] \geq Y[S \cap [t^*+2, t-2]]$.
\end{lemma}
\begin{proof}
Let $\mathcal{C}$ be the honestly held chain at slot $t$ for which chain quality is violated. Let $t^*$ be the time of mining of the last honest block in $\mathcal{C}$. Let this block be called $B$ and let its height be $\ell$. By Lemma \ref{lem:chain_growth}, we have $|\mathcal{C}| \geq \ell + Y[S \cap [t^*+2, t-2]]$. We know that the portion of $\mathcal{C}$ from height $\ell$ onward is composed only of adversarial blocks. These must be mined at or after slot $t^*$. Therefore, $|\mathcal{C}| - \ell \leq Z[t^*, t]$. This implies $Z[t^*, t] \geq Y[S \cap [t^*+2, t-2]]$. Furthermore, we know that $Z[t^*, t] \geq k$ by the condition that the last $k$ blocks do not contain any convergence opportunity.
\end{proof}

\subsection{Probabilistic guarantee}
We begin with a proposition that transforms the necessary conditions from Lemmas \ref{lem:necessary} and \ref{lem:chain_quality_necessary} to a more convenient form for analysis.

\begin{proposition}
\label{prop:necessary}
Suppose $k$-common prefix violation or $(k, s)$-chain quality violation occurs at slot $t$. Then there exist $t^* < t$ such that $Z[t^*,t] + Y[t^*,t] \geq k$ and $Z[t^*, t] \geq Y[\max({\textsf{GST},t^*})+2, t-2] - (1 + (t - t^*)/e')(d'+3)-1$.
\end{proposition}

\begin{proof}
By comparing the necessary conditions in Lemma~\ref{lem:necessary} and  Lemma~\ref{lem:chain_quality_necessary}, it suffices to show that $Z[t^*, t] \geq Y[S \cap [t^*+2, t-2]] - 1$ implies $Z[t^*, t] \geq Y[\max({\textsf{GST},t^*})+2, t-2] - (1 + (t - t^*)/e')(d'+3)-1$. Indeed, $Z[t^*, t] \geq Y[S \cap [t^*+2, t-2]] - 1$ implies
\begin{equation*}
    Z[t^*, t] \geq Y[\max({\textsf{GST},t^*})+2, t-2] - n(d'+3)-1,
\end{equation*}
where $n$ is the number of checkpoints that appear in the interval $[t^*,t]$. 
We now show a upper bound on $n$: $n \leq (t - t^*)/e' + 1$. Indeed, a block mined in slot $t^*$ cannot be checkpointed at a slot earlier than $t^*$. Of course, at slot $t$, we only see checkpoints that appear until slot $t$. No more than $\lceil (t - t^*)/e' \rceil \leq (t - t^*)/e' + 1$ checkpoints can appear over such an interval, as consecutive checkpoints are separated at least $e'$ slots apart. Thus, $n \leq (t - t^*)/e' + 1$. Therefore, $Z[t^*, t] \geq Y[\max({\textsf{GST},t^*})+2, t-2] - (1 + (t - t^*)/e')(d'+3)-1$.
\end{proof}

We conclude the proof in this section by showing that the in a typical execution, the condition of Proposition \ref{prop:necessary} does not occur after slot $t>C\cdot \textsf{GST}$ for some appropriate $C > 0$. Therefore $k$-common prefix property and $(k, s)$-chain quality property hold after slot $t>C\cdot \textsf{GST}$.

\begin{lemma}
\label{lem:prob}
For $\tau$ large enough, under ($\varepsilon,\tau$)-\textit{typical execution}, there exist a constant $C > 0$ such that $k$-common prefix property and $(k, s)$-chain quality property holds after slot $t>s$, where $k=2\lambda\Delta\tau$ and $s = C\cdot \textsf{GST}$.
\end{lemma}

\begin{proof}
Suppose $k$-common prefix property or $(k, s)$-chain quality property is violated at slot $t>C\cdot \textsf{GST}$, then by Lemma \ref{lem:necessary} we have that there exists $r<t$ such that 
\begin{equation}
\label{eqn:condition1}
Z[r,t] + Y[r,t] \geq k,    
\end{equation}
and 
\begin{equation}
\label{eqn:condition2}
Z[r, t] \geq Y[\Bar r +2, t-2] - (1 + (t - r)/e')(d'+3)-1, 
\end{equation}
where $\Bar r \triangleq \max(r, \textsf{GST})$. We only need to show that this event won't happen under typical execution. If $t-r < \frac{k}{2\lambda\Delta} = \tau$, then under typical execution we have
\begin{equation*}
    Y[r,t]+Z[r,t] \leq Y[r,r+\tau] + Z[r,r+\tau] \leq (1+\varepsilon)\lambda\Delta\tau < 2\lambda\Delta\tau = k,
\end{equation*}
which contradicts Eqn. (\ref{eqn:condition1}). Hence we have $t-r \geq \frac{k}{2\lambda\Delta} = \tau$.

Further we can rewrite Eqn. (\ref{eqn:condition2}) as
\begin{equation}
    \label{eqn:condition2_1}
    Z[r, t] \geq Y[\Bar r +2, t-2] - (t - r)d_1/e' - d_2,
\end{equation}
where $d_1 = d'+3$ and $d_2 = d'+4$.
Eqn. (\ref{eqn:condition2_1}) describes a random walk with drift. Note that the extra terms on the right side of Eqn. (\ref{eqn:condition2_1}) are small by the fact that $d_2 = O(\sqrt \kappa) = O(\sqrt \tau)$ and $e'$ can be chosen arbitrarily large. 

As we have seen $t-r \geq \tau$, then under typical execution we have
\begin{equation*}
    Z[r,t] \leq (t-r)\bar z + \varepsilon(t-r)\bar y < (1-3\varepsilon) (t-r) \bar y + \varepsilon(t-r)\bar y = (1-2\varepsilon) (t-r) \bar y.
\end{equation*}

\noindent Case 1: $r>\textsf{GST}$, i.e., $\bar r =r$. Then we have
\begin{equation*}
    Y[r+2,t-2] \geq (1-\varepsilon) (t-r-4)\bar y.
\end{equation*}
Therefore, we can always choose $e'$ to be large enough such that $(1-2\varepsilon + d_1/e') (t-r) \bar y + d_2 < (1-\varepsilon) (t-r-4) \bar y$, given that $\tau$, as well as $t-r$, is large. Therefore, $Z[r, t] < Y[\Bar r +2, t-2] - (t - r)d_1/e' - d_2$, which contradicts Eqn. (\ref{eqn:condition2_1}).

\noindent Case 2: $r\leq \textsf{GST}$, i.e., $\bar r =\textsf{GST}$. Then we have 
\begin{equation*}
    Y[\bar r+2,t-2] \geq (1-\varepsilon) (t-{\rm GST}-4)\bar y >(1-\varepsilon)(\frac{C-1}{C} t -4)\bar y.
\end{equation*}
Again we can always let $e'$ and $C$ to be large enough such that $(1-2\varepsilon + d_1/e') s \bar y + d_2 < (1-\varepsilon) (\frac{C-1}{C} t -4) \bar y$, given that $\tau$, as well as $t$, is large.
Then we have $Y[\bar r +2,t-2] > Z[0,t] + d_1t/e' + d_2 \geq Z[r,t] + d_1(t-r)/e' +d_2$, which contradicts Eqn. (\ref{eqn:condition2_1}).
\end{proof}

Then we have the following corollary, which states that the common prefix property and the chain quality property of the checkpointed longest chain  will hold after $O(\textsf{GST})$ time with high probability.

\begin{corollary}
\label{coro:common}
If we choose $k$ to be large enough such that $k = \Theta(\kappa)$, then there exist a constant $C >0$ such that $k$-common prefix property and $(k, C \cdot \textsf{GST})$-chain quality property hold after slot $C\cdot\textsf{GST}$ with probability at least $1-e^{-\Omega(\sqrt{\kappa})}$. 
\end{corollary}

\begin{proof}
All the analysis in this section conditions on that event $E$ occurs ({\bf CP1} holds), and we have $\mathbb{P}(E) \geq 1 - e^{-\Omega(\sqrt{\kappa})}$ by Corollary \ref{coro:uniform}. Then combine it with Corollary \ref{coro:typical} and Lemma \ref{lem:prob}, and apply the union bound, we can conclude the proof.
\end{proof}

\section{Discussion}\label{sec:discussion}

In this section, we elaborate on some of the finer details pertaining to our work. 

\paragraph{Choice of Algorand as the Checkpointing BFT Protocol} We chose Algorand BA as our checkpointing protocol because it satisfies the desired  properties {\bf CP0}-{\bf CP3} as listed in \S\ref{sec:protocol}. However, some other well-known BFT protocols fail to meet all these properties, and therefore we don't use them as a candidate checkpointing protocol. Generally, BFT protocols can be categorized into two classes: 1) Classic BFT protocols, such as PBFT \cite{castro1999practical} and Hotstuff \cite{yin2019hotstuff}; 2) Chained BFT protocols, such as Streamlet \cite{chan2020streamlet} and Chained Hotstuff \cite{yin2019hotstuff}. 

PBFT and Hotstuff don't satisfy the {\bf CP1} (the recency condition). The adversary could potentially lock on a block $B$ privately when it becomes the leader of one view, and then the adversary finalizes $B$ as a checkpoint after a long time when $B$ is no longer in the best chain ($B$ is a descendant of the last checkpoint but not in the longest chain). In this attack, all  honest blocks mined during this time period could be potentially wasted as they are not extending $B$. Further, this resource-bleeding attack could be turned into a perpetual liveness attack in the partial synchronous model. One simple way to fix this bug is by introducing randomness into the protocol: when a new honest leader does not hear any locked block from previous views, it will propose an empty block $\perp$ to refresh the privately locked block, with a certain probability. However, this simple fix introduces unnecessary latency in the normal path when there is no attack and the network is synchronous, as it may need more views with honest leaders to ensure progress.

On the other hand, for Streamlet and Chained Hotstuff, the finalization of one block is supported by its descendants, therefore it is hard to design a checkpointed longest chain protocol such that all finalized checkpoints lie on a single blockchain. Therefore, finding or designing a good checkpointing protocol is non-trivial, and it would be interesting to see whether there is any other BFT protocol that fits our need and perhaps has better latency and lower communication cost.

\paragraph{Checkpointing protocol properties} The finality gadget in GRANDPA \cite{stewart2020grandpa} and the finality layer of Afgjort \cite{DinsdaleYoung2020AfgjortAP} are also checkpointing protocols. The desired checkpointing properties in GRANDPA are very similar to ours: it has {\bf CP0}, {\bf CP1} and {\bf CP3}, but {\bf CP2} (gap in checkpoints) is missing. Without {\bf CP2}, the mining bleeding could happen frequently, which leads to security vulnerabilities. (Note that GRANDPA also does not have property {\bf P3}, as discussed in \S\ref{sec:related}).

Afgjort only requires {\bf CP0}, but it has two other desired properties: The ``$\Delta$-updated" property guarantees that the chains held by honest nodes are at most $\Delta$ blocks beyond the last finalized block; ``$k$-support property" ensures that any finalized block must have been on the chain adopted by at least $k$ honest nodes. The ``$\Delta$-updated" property  is in the opposite direction of our {\bf CP2}; our protocol satisfies $n/3$-support. We defer a full comparison of all these properties and identifying the minimal set of desired properties to future work.

\paragraph{Attack on GRANDPA} In \S\ref{sec:related}, we claimed that in the GRANDPA design, the $k$-deep rule is insecure in the synchronous but unsized setting. Here, we describe a possible ``roll-back" attack, using some notations from GRANDPA paper \cite{stewart2020grandpa}. Let $h$ be the number of online honest voters and we assume $f+1 \leq h \leq 2f$. Suppose at the beginning of round 1, the genesis block $B_0$ has two children $B_1$ and $B_2$ (they are siblings), where $B_1$ is an honest block and $B_2$ is a private block. Then all $h$ honest voters prevote for $B_1$ in step 3 of round 1. However, all honest voters will get stuck in step 4 if the adversary doesn't send any prevotes, because $g(V_{1,v}) = nil$ and $E_{0,v} = B_0$ hence $g(V_{1,v}) \geq E_{0,v}$ is not satisfied. However, all honest miners will mine on $B_2$ (assuming tie breaking in favor of the adversary). When $B_2$ becomes $k$-deep, the adversary casts enough prevotes and precommits on $B_1$ to make $B_1$ finalized in round 1. Now $B_2$ is deconfirmed, and all honest miners switch to mine on $B_1$, creating a {\em safety violation}. Further if the adversary has mined some private blocks on $B_1$, then he can continuous this attack in all rounds and all the finalized block will be adversarial block, leading to a {\em liveness violation}.

}
{
}

\end{document}